%% file: LMCS-2011-661_final-version.tex
\documentclass{LMCS}

\def\doi{8 (2:03) 2012}
\lmcsheading%
{\doi}
{1--29}
{}
{}
{Jul.~\phantom06, 2011}
{Apr.~26, 2012}
{}
\usepackage{latexsym,xspace}
\usepackage{amsmath}
\usepackage{amssymb}
\usepackage{stmaryrd}
\usepackage{tikz}
\usetikzlibrary{shapes}

\usepackage{algorithm,algorithmic}
\usepackage{rotating}
\usepackage{xcolor}
\usepackage{soul}

\usepackage{multirow}
\usepackage{array}

\usepackage{hyperref}
\usepackage{enumitem}

\title[Intuitionistic implication makes model checking hard]{Intuitionistic implication makes model checking hard\rsuper*}

\author[M.~Mundhenk]{Martin~Mundhenk}
\author[F~.Wei\ss{}]{Felix Wei\ss{}}
\address{Universit\"at Jena, Institut f\"ur Informatik, Jena, Germany}
\email{\{martin.mundhenk,felix.weiss\}@uni-jena.de}

 
\titlecomment{{\lsuper*}Preliminary parts of the results appeared in extended abstracts \cite{MW-RP10,MW-STACS11}.}

\newcommand{\classfont}[1]{\ensuremath{\mathsf{#1}}\xspace}

\newcommand{\NCi}{\classfont{NC^1}}
\newcommand{\ACi}{\classfont{AC^1}}
\newcommand{\Log}{\classfont{L}}
\newcommand{\NLog}{\classfont{NL}}

\newcommand{\LOGCFL}{\classfont{LOGCFL}}
\newcommand{\p}{\classfont{P}}
\newcommand{\NP}{\classfont{NP}}
\newcommand{\PSPACE}{\classfont{PSPACE}}

\newcommand{\langfont}[1]{\mbox{\textsc{#1}}\xspace}

\newcommand{\AGAP}{\langfont{Agap}}
\newcommand{\ASAGAP}{\langfont{AsAgap}}
\newcommand{\EQVF}{\langfont{EqVformula}}

\newcommand{\Lklasse}[2][]{\ensuremath{\mathrm{#2}#1}\xspace}
\newcommand{\Siv}{\Lklasse[4]{S}}     % die Modallogik S4
\newcommand{\Sivii}{\Lklasse[4.2]{S}}     % die Modallogik S4.2
\newcommand{\KK}{\Lklasse{K}}
\newcommand{\IPC}{\Lklasse{IPC}}
\newcommand{\BPL}{\Lklasse{BPL}}
\newcommand{\FPL}{\Lklasse{FPL}}

\newcommand{\KC}{\Lklasse{KC}}

\newcommand{\Kiv}{\Lklasse[4]{K}}
\newcommand{\PrL}{\Lklasse{PrL}}

\newcommand{\imodels}{\models_{\mathsf{i}}}
\newcommand{\mmodels}{\models_{\mathsf{m}}}

\newcommand{\Model}[1]{\ensuremath{\mathcal{#1}}}
\newcommand{\Klasse}[1]{\ensuremath{\mathcal{#1}}}         \newcommand{\mKlasse}[1]{\mathcal{#1}}

\newcommand{\PROP}{\operatorname{PROP}}
\newcommand{\mlogred}{\leq_{\mathrm{m}}^{\mathrm{log}}}
\newcommand{\PM}{\mathfrak{P}}

\newcommand{\topnode}{\mathit{top}}

\newlength\problemlength
\newcommand\problemdef[3]{%
\begin{list}{}{\labelwidth\problemlength \labelsep.7em \rightmargin1.5em
\leftmargin\problemlength \advance\leftmargin by3em%2em
\parsep0ex \itemsep.2ex plus.1ex}
\item[{\sl Problem:\hfill}] #1 \item[{\sl Input:  \hfill}] #2
\item[{\sl Output: \hfill}] #3
\end{list}
}
\newcommand\dproblem[3]{%
\vspace{1ex}
\begin{list}{}{\labelwidth\problemlength \labelsep.7em \rightmargin1.5em
\leftmargin\problemlength \advance\leftmargin by3em
\parsep0ex \itemsep.2ex plus.1ex}
\item[{\sl Problem:\hfill}] #1 \item[{\sl Input:  \hfill}] #2
\item[{\sl Question: \hfill}] #3
\end{list}
\vspace{1ex}
}

\newcommand{\fe}[1]{\ensuremath{\mbox{$#1$}\mbox{-\sc{\small KMc}}}}
\newcommand{\apath}{\mathit{apath}}
\newcommand{\lp}{\ensuremath{\emph{\textit{lp}}}}
\newcommand{\Top}{\ensuremath{\emph{\textit{Top}}}}
\newcommand{\eps}{\varepsilon}
\newcommand{\Klauf}{\Big{(}\hspace{3px}}
\newcommand{\Klzu}{\hspace{3px}\Big{)}}
\newcommand{\klauf}{\big{(}\hspace{2px}}
\newcommand{\klzu}{\hspace{2px}\big{)}}
\newcommand{\gklauf}{\big{\{}\hspace{2px}}
\newcommand{\gklzu}{\hspace{2px}\big{\}}}

\newtheorem{claim}{Claim}{\itshape}{\rmfamily}
\newcommand{\qedclaim}{\hspace{2mm}\rule{2mm}{2mm}}
{\itshape}{\rmfamily}

\newenvironment{mathe}{\renewcommand{\arraystretch}{1.5}\vspace{1.2ex}\begin{tabular}{>{$}l<{$}>{$}l<{$}>{$}l<{$}}}{\end{tabular}\vspace{1.2ex}\newline\renewcommand{\arraystretch}{0}}

\newcommand{\auslassen}[1]{}

\newcommand{\GT}{\mathit{gt}}

\newcommand{\iImpl}{\rightarrowtriangle}

\definecolor{mod}{HTML}{0000AA}%{33FF66}
\newcommand{\markit}[1]{{\color{mod}#1}}
\renewcommand{\markit}[1]{#1} %farbmarkierungen deaktivieren -> entferne % 

\input{Bilder-SliceGraph-Modell}

\begin{document}

%%%%%%%%%%%%%%%%%%%%%%%%%%%%%%%%%%%%%%%%%%%%%%%%%%%%%%%%%%%%%%%%%%%%%%%%%%%%%%%%%%%%%%%%%%%%%%%%%%%%%%%%%%%%
% Abstract
%%%%%%%%%%%%%%%%%%%%%%%%%%%%%%%%%%%%%%%%%%%%%%%%%%%%%%%%%%%%%%%%%%%%%%%%%%%%%%%%%%%%%%%%%%%%%%%%%%%%%%%%%%%%

\begin{abstract}
We investigate the complexity of the model checking problem for
intuitionistic and modal propositional logics \markit{over transitive Kripke models}.
More specific, we consider 
intuitionistic logic \IPC, basic propositional logic \BPL, formal propositional logic \FPL,
and Jankov's logic  \KC.
We show that the model checking problem is \p-complete for the implicational fragments of all these intuitionistic logics.
For \BPL and \FPL we reach \p-hardness even on the implicational fragment with only one variable. 
The same hardness results are obtained for the strictly implicational fragments of their modal companions.
Moreover, we investigate whether formulas with less variables and additional connectives
make model checking easier.
Whereas for variable free formulas outside of the implicational fragment, 
\FPL model checking is shown to be in \LOGCFL,
the problem remains \p-complete for \BPL.
\end{abstract}

%%%%%%%%%%%%%%%%%%%%%%%%%%%%%%%%%%%%%%%%%%%%%%%%%%%%%%%%%%%%%%%%%%%%%%%%%%%%%%%%%%%%%%%%%%%%%%%%%%%%%%%%%%%%
% Titel
%%%%%%%%%%%%%%%%%%%%%%%%%%%%%%%%%%%%%%%%%%%%%%%%%%%%%%%%%%%%%%%%%%%%%%%%%%%%%%%%%%%%%%%%%%%%%%%%%%%%%%%%%%%%

\maketitle

%%%%%%%%%%%%%%%%%%%%%%%%%%%%%%%%%%%%%%%%%%%%%%%%%%%%%%%%%%%%%%%%%%%%%%%%%%%%%%%%%%%%%%%%%%%%%%%%%%%%%%%%%%%%
% Die Abschnitte
%%%%%%%%%%%%%%%%%%%%%%%%%%%%%%%%%%%%%%%%%%%%%%%%%%%%%%%%%%%%%%%%%%%%%%%%%%%%%%%%%%%%%%%%%%%%%%%%%%%%%%%%%%%%

\input{1-Introduction}

\input{2-Preliminaries}

\input{3-intuitionistic_lower_bounds}

\input{4-modal_lower_bounds}

\input{5-Conclusion}

%%%%%%%%%%%%%%%%%%%%%%%%%%%%%%%%%%%%%%%%%%%%%%%%%%%%%%%%%%%%%%%%%%%%%%%%%%%%%%%%%%%%%%%%%%%%%%%%%%%%%%%%%%%%
% Acknowledgements
%%%%%%%%%%%%%%%%%%%%%%%%%%%%%%%%%%%%%%%%%%%%%%%%%%%%%%%%%%%%%%%%%%%%%%%%%%%%%%%%%%%%%%%%%%%%%%%%%%%%%%%%%%%%

\vspace{3ex}
\textbf{Acknowledgements.}
The authors thank Vitezslav Svejdar for helpful discussions about intuitionistic logic,
and Thomas Schneider for his support.
The authors specially thank an anonymous referee of the Workshop on Reachability Problems 2010
for her/his idea that led to Theorem~\ref{thm:S4.2_1-P-hard}.
The authors like to thank also the anonymous referees of LMCS for their helpful comments.

%%%%%%%%%%%%%%%%%%%%%%%%%%%%%%%%%%%%%%%%%%%%%%%%%%%%%%%%%%%%%%%%%%%%%%%%%%%%%%%%%%%%%%%%%%%%%%%%%%%%%%%%%%%%
% Bibliographie
%%%%%%%%%%%%%%%%%%%%%%%%%%%%%%%%%%%%%%%%%%%%%%%%%%%%%%%%%%%%%%%%%%%%%%%%%%%%%%%%%%%%%%%%%%%%%%%%%%%%%%%%%%%%

\bibliographystyle{abbrv}

%%%%%%%%%%%%%%%%%%%%%%%%%%%%%%%%%%%%%%%%%%%%%%%%%%%%%%%%%%%%%%%%%%%%%%%%%%%%%%%%%%%%%%%%%%%%%%%%%%%%%%%%%%%%
% Anhang
%%%%%%%%%%%%%%%%%%%%%%%%%%%%%%%%%%%%%%%%%%%%%%%%%%%%%%%%%%%%%%%%%%%%%%%%%%%%%%%%%%%%%%%%%%%%%%%%%%%%%%%%%%%%

\appendix
\section{}
\label{appendixA}

\noindent
\textbf{Theorem~\ref{thm:KC2-P-hard}.}~
\textit{The model checking problem for $\KC_2$ and for $\IPC_2$ is \p-hard.}

\noindent
\input{KC2_P-haerte}

\end{document}

%% file: Bilder-SliceGraph-Modell.tex
\tikzstyle{vertex}=[circle,draw=black,minimum size=20pt,inner sep=0pt]
\tikzstyle{vertex split}=[circle split,draw=black,minimum size=20pt,inner sep=0pt]
\tikzstyle{selected vertex} = [vertex, fill=black!24]
\tikzstyle{edge} = [draw,thick,->]
\tikzstyle{weight} = [font=\small]
\tikzstyle{selected edge} = [draw,line width=1.6pt,->]
\tikzstyle{dashed edge} = [draw,densely dashed,line width=1.5pt,->,black!40]
\tikzstyle{selecteddashed edge} = [draw,densely dashed,line width=1.2pt,->,black!70]
\tikzstyle{dotted edge} = [draw,densely dotted,line width=1.5pt,->,black!40]
\tikzstyle{selecteddotted edge} = [draw,densely dotted,line width=1.6pt,->,black!70]

\newcommand{\ptgraph}{%

    \foreach \pos/ \name in {{(4,0)/a2}, {(0,2)/b1}, {(6,2)/b2},
                             {(0,4)/c1}, {(6,4)/c2}, {(2,6)/d1}}
        \node[vertex] (\name) at \pos {$ $};
    \foreach \pos/ \name / \printname in {{(2,0)/a1/s}, {(4,6)/d2/t}} 
        \node[vertex] (\name) at \pos {$\printname$};
        
    \foreach \source/ \dest in {a1/b2, b2/c2, c2/d2, 
                                a2/b1, b1/c1, c1/d1}
         \path[edge] (\source) -- node[weight] {$ $} (\dest);

}

\newcommand{\ptclosuregraph}{%

    \ptgraph
        
    \foreach \source/ \dest in {a1/c1, a1/c2, a1/d1, a1/d2,
                                a2/c1, a2/c2, a2/d1, a2/d2,
                                b1/d1, b1/d2, b2/d1, b1/d2}
         \path[edge] (\source) -- node[weight] {$ $} (\dest);

}

\newcommand{\bspslicegraph}{%
    \foreach \pos/ \name in {{(0,2)/b1}, {(4,2)/b2},
                             {(0,4)/c1}, {(2,4)/c2}, {(4,4)/c3}, {(6,4)/c4},
                             {(0,6)/d1}, {(2,6)/d2},             {(6,6)/d4}}
        \node[vertex] (\name) at \pos {$ $};
    \foreach \pos/ \name / \printname in {{(0,0)/a1/s}, {(4,6)/d3/t}} 
        \node[vertex] (\name) at \pos {$\printname$};
        
    \foreach \source/ \dest in {a1/b1, b1/c1, b1/c2, b1/c3, 
                                    c1/d1, c1/d3, c2/d1, c2/d3, c3/d2, c3/d4, c4/d3}
         \path[edge] (\source) -- node[weight] {$ $} (\dest);
    \foreach \source/ \dest in {a1/b2, a1/b2, b2/c2, b2/c4,
                                c2/d3, c4/d3}
         \path[selected edge] (\source) -- node[weight] {$ $} (\dest);    
    
}
\newcommand{\slicegraphexample}{%
\begin{tikzpicture}[scale=0.6]
\small
    \bspslicegraph
    \node[xshift=-2ex](V1) at (a1.west) {$\exists$} ;
    \node[xshift=-2ex](V2) at (b1.west) {$\forall$} ;
    \node[xshift=-2ex](V3) at (c1.west) {$\exists$} ;
    \node[xshift=-2ex](V4) at (d1.west) {$\forall$} ;
    
    \node (V5) at (3,0) {Graph $G$} ;

\end{tikzpicture}
}

\newcommand{\siviiexample}{%
\begin{tikzpicture}[scale=0.6]
\small
    \bspslicegraph

    \foreach \pos/ \name / \printname in {{(0,8)/u/u}, {(3,8)/t1/t_1}, {(5,8)/t2/t_2},{(3,10)/topnode/\topnode}} 
        \node[vertex] (\name) at \pos {$\printname$};
        
    \foreach \source/ \dest in {d1/u, d2/u, d3/u, d4/u, d3/t1, d3/t2,
                                u/topnode, t1/topnode, t2/topnode} 
          \path[edge] (\source) -- node[weight] {$ $} (\dest);
     \path[edge] (t1) edge[bend left=30] (t2) ;
     \path[edge] (t2) edge[bend right=-30] (t1) ;

    \node[xshift=1.5ex] (ttt) at (t2.north east) {$a$} ;
    \node[xshift=1.5ex] (tttt) at (topnode.north east) {$a$} ;
    \node[right]  (V4) at (7,6) {all nodes in $V_4$ satisfy $a$} ;
    \node[right]  (V2) at (7,2) {all nodes in $V_2$ satisfy $a$} ;
 
\end{tikzpicture}
}

\newcommand{\kcproofexample}{%
\begin{tikzpicture}[scale=0.6]
\small
    \bspslicegraph
    \foreach \pos/ \name / \printname in {{(3,8)/topnode/\topnode}} 
        \node[vertex] (\name) at \pos {$\printname$};
        
    \foreach \source/ \dest in {d1/topnode, d2/topnode, d3/topnode, d4/topnode} 
          \path[edge] (\source) -- node[weight] {$ $} (\dest);
                   
        \foreach \source/ \dest/ \bogen in {a1/c1/40, a1/c2/0, a1/c3/0, a1/c4/5, 
                                            a1/d1/40, a1/d2/0, a1/d3/-15, a1/d4/-15,
                                            b1/d1/40, b1/d2/0, b1/d3/-18, b1/d4/-40,
                                            b2/d1/28, b2/d2/0, b2/d3/-28, b2/d4/0}
           \path[dashed edge] (\source) edge[bend left=\bogen] (\dest);
    
        \node[right]  (V5) at (6.5,8) {slice $V_5$: $a_1, a_2, a_3, a_4, a_5$} ;
        \node[right]  (V4) at (6.5,6) {slice $V_4$: $a_1, a_2, a_3$} ;
        \node[right]  (V3) at (6.5,4) {slice $V_3$: $a_1, a_2$} ;
        \node[right]  (V2) at (6.5,2) {slice $V_2$: $a_1$} ;
        \node[right]  (V1) at (6.5,0) {slice $V_1$: no variable is satisfied} ;
      
        \node[xshift=1.5ex] (ttt) at (d3.north east) {$a_4$} ;

        \node (V5) at (3.2,0) {Kripke model $\Model{M}_G$} ;
    
\end{tikzpicture}
}

\newcommand{\fpleinsexample}{%
\begin{tikzpicture}[scale=0.6]

\small
    \bspslicegraph ;
    \node[xshift=1.5ex] () at (d3.north east) {$p$} ;

    \node[right,xshift=4ex]  (V4) at (d4.east) {$\begin{array}{r@{~}c@{~}ll}
                                                     t&\imodels& p \\
                                                     u&\not\imodels& p & \text{ for $u\not=t$}
                                                \end{array}$} ;       
    \node[right,xshift=4ex]  (V3) at (c4.east) {$\begin{array}{r@{~}c@{~}ll}
                                                    g & \not\imodels & p\iImpl \bot & \text{for $g$ with $\apath_G(g,t)$}\\
                                                    u & \imodels & p \iImpl \bot & \text{for $u$ with $\mathit{not}~\apath_G(u,t)$}
                                                 \end{array}$};       
                                                 
     \node[right,xshift=4ex]  (V2) at (b2.east) {$\begin{array}{r@{~}c@{~}ll}
                                                    g & \imodels & (p\iImpl\bot)\iImpl(\top\iImpl \bot) & \text{for $g$ with $\apath_G(g,t)$}\\
                                                    u & \not\imodels & (p\iImpl\bot)\iImpl(\top\iImpl \bot) & 
                                                                                                   \text{for $u$ with $\mathit{not}~\apath_G(u,t)$}
                                                 \end{array}$};       
                                                 
     \node[right,xshift=20ex]  (V1) at (a1.east) {$\begin{array}{r@{~}c@{~}ll}
                                                    s & \not\imodels & ((p\iImpl\bot)\iImpl(\top\iImpl \bot)) 
                                                                    \iImpl (\top \iImpl (\top \iImpl \bot)) 
                                                 \end{array}$};

\end{tikzpicture}
}

%% file: 1-Introduction.tex
%%%%%%%%%%%%%%%%%%%%%%%%%%%%%%%%%%%%%%%%%%%%%%%%%%%%%%%%%%%%%%%%%%%%%%%%%%%%%%%%%%%%%%%%%%%%%%%%%%%%%%%%%%%%
% Introduction
%%%%%%%%%%%%%%%%%%%%%%%%%%%%%%%%%%%%%%%%%%%%%%%%%%%%%%%%%%%%%%%%%%%%%%%%%%%%%%%%%%%%%%%%%%%%%%%%%%%%%%%%%%%%

\section{Introduction}

%% the general question
Intuitionistic propositional logic \IPC (see e.g.~\cite{dal04}) goes back to Heyting
and bases on Brouwer's idea of constructivism from the beginning of the 20th century.
It can be seen as the part of classical propositional logic
that goes without the use of the excluded middle $a\vee \neg a$.

While it was originally conceived and is primarily of interest from a proof-theoretic point of view, \IPC admits many sound and complete semantics, such as the algebraic semantics \cite{Tar38}, the topological semantics \cite{McKinseyTarski44}, and the arithmetical semantics \cite{DeJonghJSL1970}. 
The most well known semantics for IPC is Kripke's possible world semantics \cite{Kripke65}. 
As a matter of fact, already in the 1930s it was observed by G\"odel %\cite{Goedel33} 
that \IPC can be mapped to a fragment of the modal logic \Siv, 
which was later shown to be the  modal logic of the class of transitive and reflexive Kripke frames \cite{Kripke63ML}.
In this paper, we explore this Kripke semantics further. 

%\markit{The most common semantics for intuitionistic logic are Heyting semantics \cite{Heyt71} 
%and Kripke semantics \cite{kri63a,kri63b}---see also \cite[Chap. 2]{SOUR06}.
%The Heyting semantics bases on Heyting algebras,
%and Kripke semantics bases on directed graphs that can straightforwardly be adapted
%to model state-transition systems.
%Therefore it is used as the standard semantics for modal and hybrid logics.
%The model checking problem 
%is the question whether a formula is satisfied in a state of a Kripke model.
%Model checking as we do with Kripke models is not suited for Heyting algebras.
%Therefore we also use Kripke semantics for intuitionistic logic.
%All our and all mentioned complexity results below refer to Kripke semantics.}

Whereas the complexity of the validity problem for \IPC is deeply studied~\cite{stat79,svejdar03,Rybakov06,JH93},
the exact complexity of its model checking problem is open.
\markit{Research on the complexity of model checking on Kripke models
goes back to \cite{filad79,sicl85} (where it is called determination of truth) and has been done for a variety of logics
like dynamic logic and many temporal} \markit{logics.}
%
%\textbf{There are especially for model checking in temporal logics many results \cite{schnoeb02,AM11diss} known. 
%For \CTL see \cite{lamasc01,CTL-MC11} and for \LTL see \cite{bamuscscscvo07entcs}.
%Also for some intuitionistic logics the model checking problem was investigated.} %\cite{MW-STACS11,MW10b,MW-RP10}.}
%
It was recently shown that the model checking problem for \IPC formulas with one variable
is \ACi-complete \cite{MW-STACS11}.
We investigate the complexity of model checking for different intuitionistic logics
and for related modal logics---their modal companions.
Our central question is which ingredients (i.e.~logical connectives, number of variables)
are needed in order to obtain maximal hardness of the model checking problem.

%% more details on the intuitionistic logics
We consider the intuitionistic logics \BPL (basic propositional logic \cite{visser80}), 
\FPL (formal propositional logic \cite{visser80}),
\IPC and % the superintuitionistic logic 
\KC (Jankov's logic, see~\cite{DL59}).
All have semantics that is defined over Kripke models with 
a monotone valuation function and a transitive frame\markit{\footnote{\markit{Unless otherwise stated we expect in the following every Kripke model to be transitive.}}} (as for \BPL) that distinguish 
on whether the frame is additionally irreflexive (\FPL), reflexive (\IPC), or a directed preorder (\KC).
The validity problem for all these logics is \PSPACE-complete~\cite{stat79,Chagrov85,svejdar03},
%% even for their implicational fragments~\cite{stat79,Chagrov85,svejdar03}, %% kommt spaeter
and the satisfiability problem is \NP-complete for \IPC and for \KC,
but in \NCi for \BPL and for \FPL.
These intuitionistic logics can be embedded into the modal logics
\Kiv, \PrL (provability logic \markit{\cite{bo93,sve00}}), \Siv, and \Sivii,
that are called the modal companions of the respective intuitionistic logic.
The validity problem and the satisfiability problem 
is \PSPACE-complete for all these modal logics~\cite{lad77,spaan93}.
\markit{The \PSPACE-completeness results mentioned 
also hold for the implicational fragment of intuitionistic logics~\cite{stat79,Chagrov85,svejdar03}
resp.~the strictly implicational fragment for the considered modal logics~\cite{Bou04}.}
Also, the complexity of the validity problem for
fragments of the considered logics with a bounded number of variables was investigated~\cite{svejdar03a,ChagrovR02,Rybakov06}.
Roughly speaking, the number of variables that is needed to obtain
a \PSPACE-hard validity problem
depends on whether the semantics restricts the transitive frames (of the Kripke models)
to be reflexive, irreflexive, or none of both.
For intuitionistic logics, it is shown in~\cite{Rybakov06} that
on transitive and reflexive frames (\IPC) one needs two variables to reach \PSPACE-hardness for the validity problem,
on transitive and irreflexive frames (\FPL) one variable is necessary,
and on arbitrary transitive frames (\BPL) one comes out without variables at all.
For their modal companions,
the same bounds apply for transitive and irreflexive frames (\PrL)~\cite{svejdar03a} 
and for arbitrary transitive frames (\Kiv)~\cite{ChagrovR02},
but for transitive and reflexive frames (\Siv) already one variable suffices~\cite{ChagrovR02}.
Notice that no \PSPACE-hardness results are known for the implicational fragment
with a bounded number of variables.

%% model checking
The model checking problem is the following decision problem.
Given a formula, a Kripke model, and a state in this model, 
decide whether the formula is satisfied in that state.
For classical propositional logic, the model checking problem 
(also called the formula evaluation problem) can be solved
% in logarithmic space~\cite{Lynch77} and even better 
in alternating logarithmic time~\cite{bus87}.
Since the models for classical propositional logic can be seen as a special case of Kripke models 
that consist of only one state,
we cannot expect such a low complexity for intuitionistic logics,
where the models may consist of many states.
For the considered logics, the upper bound \p follows from~\cite{filad79}.
In fact, this upper bound turns out to be the lower bound too---%
we show that the model checking problem for \KC, \IPC, \BPL, and \FPL is \p-complete,
even on the implicational fragments.
We obtain the same bounds on the number of variables
for the \p-hardness of the model checking problem as for the \PSPACE-hardness
of the validity problem (see above) for the considered intuitionistic logics
and their modal companions.
Other than for the validity problem, we obtain \p-hardness
even on the implicational fragments of \FPL and \BPL with one variable.
The \PSPACE-hardness of the validity problem on these fragments is open.
Since the implicational fragments of \IPC and \KC
for any bounded number of variables have only a finite number
of equivalence classes (see~\cite{Urq74}),
we cannot expect to get \p-hardness of model checking on these fragments.
%
%Our hardness results translate to the strictly implicational
%fragments of the modal companions---i.e.~the fragment of formulas
%that have the strict implication $\Box(\cdot \rightarrow \cdot)$ as only connective. 
We also consider optimality of the \p-hardness results in the sense
whether model checking with less variables has complexity below \p.
We show that model checking for the variable free fragment of \FPL drops to \LOGCFL,
whereas for \BPL one can trade the variable in an $\vee$ and keeps \p-hardness.

Our results base on a technique we use to show
that the model checking problem for the implicational fragment
of \IPC is \p-hard.
The variables we use in our construction are essentially needed
to measure distances in the model and to mark a certain state.
In order to restrict the use of variables, it suffices 
to express these in a different way.
This takes different numbers of variables in the different logics
according to their frame properties.
  
%\enlargethispage*{5mm}

This paper is organized as follows.
In Section~\ref{sec:prelims} we introduce the notations
for the logics under consideration,
and we show \p-completeness of a 
graph accessibility problem 
for a special case of alternating graphs
that will be used for our \p-hardness proofs.
In Section~\ref{sec:intuitionistic lower bounds} 
we consider model checking for the intuitionistic logics \KC, \IPC, \FPL, and \BPL.
It starts with the \p-hardness results (Section~\ref{subsec:p-hard}),
and closes with the optimality of bounds on the number of variables needed to obtain \p-hardness (Section~\ref{subsec:optimality}).
In Section~\ref{sec:modal lower bounds}
the results for the modal companions \Sivii, \Siv, \PrL, and \Kiv follow.
The arising completeness results and conclusions are drawn in Section~\ref{sec:conclusion}.
An overview of the results is given in Figures~\ref{fig:IPC_results} and \ref{fig:ML_results}.

%% file: 2-Preliminaries.tex
%%%%%%%%%%%%%%%%%%%%%%%%%%%%%%%%%%%%%%%%%%%%%%%%%%%%%%%%%%%%%%%%%%%%%%%%%%%%%%%%%%%%%%%%%%%%%%%%%%%%%%%%%%%%
% Preliminaries
%%%%%%%%%%%%%%%%%%%%%%%%%%%%%%%%%%%%%%%%%%%%%%%%%%%%%%%%%%%%%%%%%%%%%%%%%%%%%%%%%%%%%%%%%%%%%%%%%%%%%%%%%%%%

\section{Preliminaries}
\label{sec:prelims}

\subsection*{Kripke Models.} 
We will consider different propositional logics
whose formulas base on a countable set $\PROP$ of \textit{propositional variables}.
A \emph{Kripke model} is a triple $\Model{M} = (U,R,\xi)$,
where $U$ is a nonempty and finite set of \textit{states}, 
$R$ is a binary relation on $U$,
and $\xi: \PROP \to \PM(U)$ is a function\,---\,the \textit{valuation function}.
For any variable it assigns the set of states
in which this variable is satisfied.
$(U,R)$ can also be seen as a directed graph---it is called a \emph{frame} in this context.
A frame $(U,R)$ is \emph{reflexive}, if $(x,x) \in R$ for all $x\in U$,
it is \emph{irreflexive}, if $(x,x)\notin R$ for all $x\in U$,
and it is \emph{transitive}, if for all $a,b,c\in U$, it follows from $(a,b) \in R$ and $(b,c) \in R$ that $(a,c) \in R$.
A reflexive and transitive frame is called a \emph{preorder}.
If a preorder $(U,R)$ has the additional property that
for all $a,b\in U$ there exists a $c\in U$ with $(a,c)\in R$ and $(b,c) \in R$,
then $(U,R)$ is called a \emph{directed preorder}.

\subsection*{Modal Propositional Logic.} 
The language \Klasse{ML} of modal logic is the set of all formulas of the form

\begin{mathe}
	\varphi & ::= & \bot ~\mid~ p ~\mid~ \varphi\rightarrow\varphi ~\mid~ \Box \varphi\,,
\end{mathe}
%
%{\centering
%
%$\varphi ::= \bot ~\mid~ p ~\mid~ \varphi\rightarrow\varphi ~\mid~ \Box \varphi\,,$
%
%}
%
%\noindent
where $p \in \PROP$.
As usual, we use the abbreviations    $\neg\varphi := \varphi\rightarrow\bot$, $\top := \neg \bot$,
$\varphi\lor\psi := (\neg\varphi)\rightarrow\psi$, $\varphi\land\psi := \neg (\varphi \rightarrow\neg\psi)$, 
and $\Diamond\varphi := \neg \Box \neg \varphi$.

The semantics is defined via Kripke models. 
Given a Kripke model $\Model{M} = (U,R,\xi)$ and a state $s \in U$, the
\textit{satisfaction relation for modal logics} $\mmodels$ is defined as follows.

\begin{mathe}
	\Model{M},s \not\mmodels \bot & & \\
	\Model{M},s \mmodels p & \text{~~iff~~} & s \in \xi(p),~ p\in\PROP, \\
	\Model{M},s \mmodels \varphi\rightarrow\psi & \text{~~iff~~} & \Model{M},s \not\mmodels \varphi \text{~or~} \Model{M},s \mmodels \psi, \\
	\Model{M},s \mmodels \Box\varphi & \text{~~iff~~} & \forall t\in U \text{ with } (s,t) \in R : \Model{M},t \mmodels \varphi.
\end{mathe}
%
%    \begin{align*}
%      \Model{M},s & \not\mmodels \bot            \\
%      \Model{M},s & \mmodels p               & \text{if and only if} & & &s \in \xi(p),~ p\in\PROP,      \\
%      \Model{M},s & \mmodels \varphi\rightarrow\psi & \text{if and only if}  & & & \Model{M},s \not\mmodels \varphi
%                                                \text{~or~} \Model{M},s \mmodels \psi,                    \\
%      \Model{M},s & \mmodels \Box\varphi & \text{if and only if} & & &
%                                      \forall t\in U: \text{ if } sRt \text{~then~} \Model{M},t \mmodels \varphi .
%    \end{align*}
%    
For $\mKlasse{M},s \mmodels \varphi$
we say that
formula $\varphi$ is \textit{satisfied} by model $\mKlasse{M}$ in state $s$.
%% If it is satisfied by $\mKlasse{M}$ in every state $s$ of $\mKlasse{M}$, then we write $\mKlasse{M}\mmodels \varphi$.

The modal logic defined in this way is called \Lklasse{K} % (after Saul Kripke) % (or Knowledge) 
and it is the weakest normal modal logic.
We will consider the stronger modal logics  \Kiv, \Siv, \Sivii, and \PrL.
The formulas in all these logics are the same as for \Klasse{ML}.
Since we are interested in model checking,
we use the semantics defined by Kripke models.
They will be defined by properties of the frame $(U,R)$ that is part of the model.
The semantics of \Kiv is defined by transitive frames.
This means, that a formula $\alpha$ is a theorem of \Kiv if and only if
$\mKlasse{M},w \mmodels \alpha$ for all Kripke models $\mKlasse{M}$ whose frame is transitive and all states $w$ of $\Model{M}$.
The semantics of \Siv is defined by preorders,
of \Sivii by directed preorders,
and of \PrL by transitive and irreflexive frames.

\subsection*{Intuitionistic Propositional Logic.}
The language \Klasse{IL} of intuitionistic logic 
is essentially the same as that of classical propositional logic, i.e. it is the set of all formulas of the form

\begin{mathe}
	\varphi & ::= & \bot ~\mid~ p ~\mid~ \varphi\land\varphi ~\mid~  \varphi\lor\varphi ~\mid~ \varphi\iImpl\varphi \,,
\end{mathe}
where $p \in \PROP$.
As usual, we use the abbreviations  $\neg\varphi := \varphi\iImpl\bot$ and $\top := \neg \bot$.
Because of the semantics of intuitionistic logic, 
one cannot express $\land$ or $\lor$ using implication and $\bot$.
\markit{Therefore we use $\iImpl$ instead of $\rightarrow$.}
%With $\Klasse{IL}_i$ we denote the subset of \Klasse{IL} that constists of all formulas with at most $i$ variables.

The semantics is defined via {Kripke models} $\Model{M} = (U,\triangleleft\,,\xi)$ that fulfil certain restrictions.
Firstly, $\triangleleft$ is transitive,
and secondly, the valuation function $\xi: \PROP \to \PM(U)$ is monotone in the sense that
for every $p\in \PROP$, $a,b\in U$: if $a\in \xi(p)$ and $a\triangleleft b$, then $b\in \xi(p)$.
We will call models that fulfil both these properties \emph{intuitionistic} or \emph{model for \BPL}.
An intuitionistic model $\Model{M}=(U,\triangleleft,\xi)$ where $\triangleleft$ is additionally reflexive (i.e.~$\triangleleft$ is a preorder)
is called a \emph{model for \IPC}.
If $\triangleleft$ is a directed preorder, then $\Model{M}$ is called a \emph{model for \KC},
and if $\triangleleft$ is irreflexive, $\Model{M}$ is called a \emph{model for \FPL}.

Given an intuitionistic model $\Model{M} = (U,\triangleleft\,,\xi)$ and a state $s \in U$, the
\textit{satisfaction relation for intuitionistic logics} $\imodels$ is defined as follows.

\begin{mathe}
	\Model{M},s \not\imodels \bot & & \\
	\Model{M},s \imodels p & \text{~~iff~~} & s \in \xi(p),~ p\in\PROP,\\
	\Model{M},s \imodels \varphi\land\psi & \text{~~iff~~} & \mKlasse{M},s \imodels \varphi \text{~and~} \mKlasse{M},s \imodels \psi, \\
	\Model{M},s \imodels \varphi\lor\psi & \text{~~iff~~} & \mKlasse{M},s \imodels \varphi \text{~or~} \mKlasse{M},s \imodels \psi, \\
	\Model{M},s \imodels \varphi\iImpl\psi & \text{~~iff~~} & \forall n\in U \text{ with } s\triangleleft n:                             \text{~if~} \mKlasse{M},n \imodels \varphi \text{~then~} \mKlasse{M},n \imodels \psi
\end{mathe}
%
%    \begin{alignat*}{2}
%      \Model{M},s & \not\imodels \bot            \\
%      \Model{M},s & \imodels p               & & \text{~~iff~~} s \in \xi(p),~ p\in\PROP,            \\
%      \Model{M},s & \imodels \varphi\land\psi & & \text{~~iff~~} \mKlasse{M},s \imodels \varphi
%                                                \text{~and~} \mKlasse{M},s \imodels \psi,                    \\
%      \Model{M},s & \imodels \varphi\lor\psi & & \text{~~iff~~} \mKlasse{M},s \imodels \varphi
%                                                \text{~or~} \mKlasse{M},s \imodels \psi,                    \\
%      \Model{M},s & \imodels \varphi\rightarrow\psi & & \text{~~iff~~} 
%                                                \forall n\in U \text{ with } s\triangleleft n: 
%                              \text{~if~} \mKlasse{M},n \imodels \varphi \text{~then~} \mKlasse{M},n \imodels \psi
%    \end{alignat*}
% 
An important property of intuitionistic logic is that the monotonicity property of the valuation function
also holds for all formulas $\varphi$:
if $\Model{M},s \imodels \varphi$ then $\forall n$ with $s\triangleleft n$ holds $\Model{M},n \imodels \varphi$.

A formula $\varphi$ is \textit{satisfied} by an intuitionistic model $\Model{M}$ in state $s$ if and only if $\Model{M},s \imodels \varphi$. 
Basic propositional logic \BPL \cite{visser80} (resp. \IPC, \KC, \FPL \cite{visser80}) is the set of \Klasse{IL}-formulas
that are satisfied by every model for \BPL (resp. \IPC, \KC, \FPL) in every state. 

\subsection*{Modal Companions.} 
G\"odel-Tarski translations % (see e.g. \cite[p.96]{MLbook}) 
map intuitionistic formulas to modal formulas in a way
that preserves validity in the different logics.
We take the translation $1$ from \cite{visser80},
that we call $\GT$ and that is defined as follows.

\begin{mathe}
	\GT(\bot) & := & \bot \\
	\GT(p) & := & p \wedge \Box p ~~~~~~~(\text{for all } p\in\PROP) \\
	\GT(\alpha \wedge \beta) & := & \GT(\alpha)\wedge\GT(\beta) \\
	\GT(\alpha \vee \beta) & := & \GT(\alpha)\vee\GT(\beta) \\
	\GT(\alpha\iImpl\beta) & := & \Box(\GT(\alpha)\rightarrow\GT(\beta))
\end{mathe}
%
%\begin{itemize}
%\item $\GT(\bot) = \bot$
%\item $\GT(p)=p \wedge \Box p$ \ ~~~~~~~\ (for all $p\in\PROP$)
%\item $\GT(\alpha \wedge \beta) = \GT(\alpha)\wedge\GT(\beta)$
%\item $\GT(\alpha \vee \beta) = \GT(\alpha)\vee\GT(\beta)$
%\item $\GT(\alpha\rightarrow\beta) = \Box(\GT(\alpha)\rightarrow\GT(\beta))$
%\end{itemize}
Visser~\cite{visser80} showed that
$\alpha$ is valid in \FPL if and only if $\GT(\alpha)$ is valid in \PrL.
Therefore, \PrL is called \emph{modal companion} of \FPL.
It is straightforward to see that $\GT$ can also be used to show that
\Kiv (resp. \Siv, \Sivii) is a modal companion of \BPL (resp. \IPC, \KC).
Figure~\ref{fig:int-modal} gives an overview
of the intuitionistic logics and their modal companions used here.
% Moreover, the G\"odel-Tarski translation $\GT$ also preserves satisfaction in the different logics.

\begin{figure}[t]

\centering

\begin{tabular}{|ccc|}
\hline
\rule{0mm}{4mm}
 intuitionistic logic & modal companion &  frame properties \\ \hline
\rule{0mm}{4mm}%
\BPL	   &  \Kiv           & transitive \\
\IPC       &  \Siv           & transitive and reflexive (= preorder)\\
\KC        &  \Sivii         & directed preorder \\
\FPL       &  \PrL           & transitive and irreflexive \\ 
%\PC        &  \Sv            & equivalence relation \\ 
\hline
\end{tabular}
\caption{Intuitionistic logics, their modal companions, and the common frame properties.}
\label{fig:int-modal}
\end{figure}

\subsection*{Model Checking Problems.}
    This paper examines the model checking problems \fe{L} for logics $L$
    whose formulas are evaluated on Kripke models with different properties.
  \dproblem{\fe{L}}{$\langle \varphi,\Model{M},s \rangle$,  where \\ $\varphi$ is a formula for $L$, 
              $\Model{M} = (U, R, \xi)$ is a Kripke model for $L$,
      and $s \in U$}{Is $\varphi$ satisfied by $\Model{M}$ in state $s$?}

\noindent We assume that formulas and Kripke models are encoded
in a straightforward way.
This means, a formula is given as a text,
and the graph $(U,R)$ of a Kripke model
is given by its adjacency matrix that takes $|U|^2$ bits.
Therefore, only finite Kripke models can be considered
%
%Notice that all instances $\langle \varphi,\Model{M},s \rangle$ of \fe{\IPC}
%have a graph $(U,R)$ contained in $\Model{M}$ that is a preorder.
%Instances without this property can be assumed to be rejected.
%The same holds for \fe{\Siv} and $\fe{\Siv_1}$.
%Accordingly, \fe{\KC}, \fe{\Sivii}, and $\fe{\Sivii_1}$ (resp. \fe{\LC} and \fe{\Siviii})
%have instances only where the graph underlying the model is a directed preorder (resp. linear preorder).
%\marginpar{\rule{2mm}{20mm}}
%Since we only consider finite models,
%every directed preorder must have a maximal element.
and it can be easily decided whether the model has the order property for the logic under consideration.

\subsection*{Complexity.}
We assume familiarity with the standard notions of complexity theory as, e.\,g., defined in \cite{pap94}.
% In particular, we will show results for the classes \DLOGCFL  and \p.
The complexity classes we use in this paper are \p (polynomial time) and some of its subclasses.
\LOGCFL is the class of sets that are \markit{logspace many-one reducible} to context-free languages.
It is also characterized as sets decidable by a nondeterministic Turing machine in polynomial time and logarithmic space
with additional use of a stack.
%\LOGCFL (resp. \DLOGCFL) is the class of sets that are logspace many-one reducible to context-free languages
%(resp. deterministic context-free languages).
%It is also characterized as sets decidable by a nondeterministic (resp. deterministic) Turing machine in polynomial time and logarithmic space
%with additional use of a stack.
$\Log$ denotes logspace, and $\NLog$ nondeterministic logspace.
%The relation between $\NLog$ and $\DLOGCFL$ is unknown.
To round off the picture,
\NCi (= alternating logarithmic time) is the class for which
the model checking problem for classical propositional logic is complete \cite{bus87},
and the model checking problem for $\IPC_1$ is complete for
\ACi (= alternating logspace with logarithmically bounded number of alternations) \cite{MW-STACS11}.
The inclusion structure of the classes under consideration is as follows.

$$
\NCi ~~\subseteq~~ \Log ~~\subseteq~~ \NLog ~~\subseteq~~ \LOGCFL ~~\subseteq~~ \ACi ~~\subseteq~~ \p 
$$

%$$
%\NCi ~~\subseteq~~ \Log 
%         ~~\begin{array}[c]{ccc}
%               & \LOGDCFL & \begin{turn}{-20}$\subseteq$\end{turn} \\[-2ex]
%                  \begin{turn}{18}$\subseteq$\end{turn}& &  \\[1ex]
%               & \NLog & \begin{turn}{20}$\subseteq$\end{turn} \\[-3.7ex]             
%                 \begin{turn}{-18}$\subseteq$\end{turn} \\[2ex]       
%         \end{array}~~
%\LOGCFL ~~\subseteq~~ \ACi ~~\subseteq~~ \p 
%$$
Fisher and Ladner~\cite{filad79} showed that model checking for modal logic is in \p.

\begin{thm}{\cite{filad79}}\label{thm:FisherLadner}
$\fe{\KK}$ is in \p. \qed
\end{thm}

The notion of reducibility we apply is the logspace many-one reduction $\mlogred$.
The G\"odel-Tarski translation $\GT$ can be seen as such a reduction between the model checking problems
for intuitionistic logics and their modal companions, namely 
 $\fe{\BPL}\mlogred \fe{\Kiv}$,
 $\fe{\IPC}\mlogred \fe{\Siv}$,
 $\fe{\KC}\mlogred \fe{\Sivii}$, and
 $\fe{\FPL}\mlogred \fe{\PrL}$.
Since $\GT$ does not introduce additional variables, the respective reducibilities also hold for the model checking problems
for formulas with any restricted number of variables.
It therefore follows from Theorem~\ref{thm:FisherLadner}
that \p is an upper bound for all model checking problems
for modal respectively intuitionistic logics considered in this paper.

\subsection*{Fragments of Logics.}
We consider fragments with bounded number of variables or $\iImpl$ as only connective.
The implicational formulas are the formulas with $\iImpl$ and $\bot$ as only connectives.
For an intuitionistic logic $L$, 
we use $L^{\iImpl}$ to denote the implicational formulas of $L$, i.e.~its \emph{implicational fragment}.  
$L_i$ denotes its \emph{fragment with $i$ variables}, i.e.~the formulas of $L$ with at most $i$ variables.
$L^{\iImpl}_i$ denotes the implicational fragment with $i$ variables.
For modal logics, the \emph{(strictly) implicational fragment} consists of formulas of the form

\begin{mathe}
	\varphi & ::= & \bot ~\mid~ p ~\mid~ \Box(\varphi\rightarrow\varphi)\,.
\end{mathe}
We use the same notation for implicational fragments of modal logics (resp. with bounded numbers of variables)
as for intuitionistic logics.

% The \emph{purely implicational fragment} is the implicational fragment without formulas that contain $\bot$. 

The G\"odel-Tarski translation $\GT$ does not translate formulas of the implicational fragment
of intuitionistic logics into the strictly implicational fragment of modal logics.
For the model checking problem,
we can use a different translation that preserves satisfaction but does not preserve validity.
Let $\GT'$ be the translation that is the same as $\GT$
but $\GT'(p)=p$ for every variable $p$. 

\begin{lem}\label{lem:Goedel-Tarski}
Let $\alpha$ be an \Klasse{IL}-formula,
and $\Model{M}$ be an intuitionistic model with state $s$.
Then $\Model{M},s\imodels \alpha$ if and only if $\Model{M},s\mmodels\GT'(\alpha)$.
If $\alpha$ is an implicational formula, then $\GT'(\alpha)$ is strictly implicational. \qed
\end{lem}

\vspace{-1.5ex}
\subsection*{\texorpdfstring{\p-complete}{P-complete} Problems.}
Chandra, Kozen, and Stockmeyer \cite{chkost81}
have shown that the Alternating Graph Accessibility Problem \AGAP
is \p-complete.
In \cite{grhoru95} it is mentioned that \p-completeness also
holds for a bipartite version.

An \emph{alternating graph} $G=(V,E)$ is a bipartite directed graph where $V=V_{\exists} \cup V_{\forall}$
are the partitions of $V$.
Nodes in $V_{\exists}$ are called \emph{existential} nodes,
and nodes in $V_{\forall}$ are called \emph{universal} nodes.
The property $\apath_G(x,y)$ for nodes $x,y\in V$ expresses that there
exists an alternating path through $G$ from node $x$ to node $y$,
and it is defined as follows.

\vspace{5px}
\begin{itemize}[itemsep=3pt]
	\item[\textnormal{1)}] $\apath_G(x,x)$ holds for all $x\in V$
	\item[\textnormal{2a)}] for $x\in V_{\exists}$: $\apath_G(x,y)$ ~if and only if~ $\exists z\in V_{\forall}: 
                                                      (x,z)\in E \text{ and } \apath_G(z,y)$
	\item[\textnormal{2b)}] for $x\in V_{\forall}$: $\apath_G(x,y)$ ~if and only if~ $\forall z\in V_{\exists}: 
                                          \text{ if } (x,z)\in E \text{ then } \apath_G(z,y)$\medskip
\end{itemize}

\noindent The problem \AGAP consists
of directed bipartite graphs $G$ and nodes $s,t$ that satisfy the property $\apath_G(s,t)$.
Notice that in bipartite graphs existential and universal nodes are strictly alternating.

\dproblem{\AGAP}{$\langle G,s,t \rangle$, where $G$ is a directed bipartite graph}{%
                 does $\apath_G(s,t)$ hold?}

\vspace{-1ex}
\begin{thm}\cite{chkost81,grhoru95}
\AGAP is \p-complete. \qed % under $\mlogred$-reductions.
\end{thm}

For our purposes, we need an even more restricted variant of \AGAP.
\markit{We require that} the graph is \emph{sliced}.
An \emph{alternating slice graph} $G=(V,E)$ is
a directed bipartite acyclic graph with a bipartitioning $V=V_{\exists} \cup V_{\forall}$,
and a further partitioning
$V=V_1 \cup V_2 \cup \cdots \cup V_m$ ($m$ \emph{slices}, $V_i\cap V_j=\emptyset$ if $i\not=j$)
where 

\begin{mathe}
	V_{\exists} & = & \bigcup\limits_{i\leq m, i \text{ odd}} V_i, \\
	V_{\forall} & = & \bigcup\limits_{i\leq m, i \text{ even}} V_i, \text{ and} \\
	E & \subseteq & \bigcup\limits_{i=1,2,\ldots,m-1} V_i \times V_{i+1}, ~~~\text{ i.e. all edges go from slice } V_i \text{ to slice }V_{i+1}.\\
%		& & ~~~\text{ i.e. all edges go from slice } V_i \text{ to slice }V_{i+1} ~~(\text{for } i=1,2,\ldots,m-1).
\end{mathe}	
%
%\begin{itemize}
%\item $V_{\exists}=\bigcup\limits_{i\leq m, i \text{ odd}} V_i$,
%\item $V_{\forall}=\bigcup\limits_{i\leq m, i \text{ even}} V_i$, and
%\item $E\subseteq \bigcup\limits_{i=1,2,\ldots,m-1} V_i \times V_{i+1}$, \\
%      i.e. all edges go from slice $V_i$ to slice $V_{i+1}$ (for $i=1,2,\ldots,m-1$).
%\end{itemize}
Finally, \markit{we require that} all nodes in a slice graph excepted those in the last slice $V_m$ 
have outdegree $>0$.

\dproblem{\ASAGAP}{$\langle G,s,t \rangle$, where $G=(V_{\exists}\cup V_{\forall},E)$ is a slice graph
                                                  with slices $V_1,\ldots,V_m$,
                                          and $s\in V_1\cap V_{\exists}$, $t\in V_m\cap V_{\forall}$}{%
                 does $\apath_G(s,t)$ hold?}

\noindent It is not hard to see that this version of the
alternating graph accessibility problem remains \p-complete.

\begin{lem}\label{lemma:ASAGAP-P-complete}
\ASAGAP is \p-complete. % under $\mlogred$-reductions.
\end{lem}

\noindent
\textit{Proof sketch.}
\ASAGAP is in \p, since it is a special case of \AGAP.
%that is known to be in \p, and since instances $\langle G,s,t \rangle$ 
%where $G$ is not a slice graph or $s\not\in V_1\cap V_{\exists}$ or $t\not\in V_m\cap V_{\forall}$
%can easily be identified.
%
In order to show \p-hardness of \ASAGAP, it suffices to find a reduction
$\AGAP \mlogred \ASAGAP$.
For an instance $\langle G,s,t \rangle$ of $\AGAP$ with graph $G=(V.E)$ where $V=V_{\exists} \cup V_{\forall}$ has $n$ nodes,
we construct an alternating slice graph $G'=(V',E')$ with $m=2n$ slices as follows.
Let $V'_i=\{\langle v,i\rangle \mid v\in V\}$ for $1\leq i\leq m$,
$V'_{\exists} = \bigcup_{i \text{~odd}} V'_i$, and $V'_{\forall}= \bigcup_{i \text{~even}} V'_i$.
The edges outgoing from a slice $V_i$ for odd $i<n$ (existential slice) are
$$
  E'_i=\Bigl\{\bigl(\langle u,i \rangle, \langle v,i+1 \rangle \bigr) \Bigm | (u,v)\in E \text{~and~} u\in V_{\exists}-\{t\}\Bigr\}
      \cup \Bigl\{\bigl(\langle u,i \rangle, \langle u,i+1 \rangle \bigr) \Bigm | u\in V_{\forall} \cup \{t\}\Bigr\}
$$
and for even $i$ (universal slice) accordingly
$$
  E'_i=\Bigl\{\bigl(\langle u,i \rangle, \langle v,i+1 \rangle \bigr) \Bigm | (u,v)\in E \text{~and~} u\in V_{\forall}-\{t\}\Bigr\}
      \cup \Bigl\{\bigl(\langle u,i \rangle, \langle u,i+1 \rangle \bigr) \Bigm | u\in V_{\exists} \cup \{t\}\Bigr\}~~.
$$
Then $G'=(V'_{\exists}\cup V'_{\forall}, ~E'_1\cup\cdots\cup E'_{m-1})$.
The transformation from $G$ to $G'$ can be computed in logarithmic space.
It is not hard to see that
$\langle G,s,t \rangle\in \AGAP$ if and only if $\langle G',\langle s,1 \rangle,\langle t,m\rangle\rangle\in\ASAGAP$.
\qed

Our basic \p-hardness proofs of model checking problems will use logspace reductions from $\ASAGAP$.
The structural basis can be seen in the proof of the folklore result
about $\KK_0$---the fragment of modal logic without variables---\markit{that we extend
to the strictly implicational fragment $\KK^{\rightarrow}_0$}.

\begin{thm}\label{thm:K-P-hard}
\markit{The model checking problem for $\KK^{\rightarrow}_0$ is \p-hard.}
\end{thm}

\proof
First, we give a straightforward transformation from $\ASAGAP$ to $\fe{\KK_0}$.
Second, we turn this into a reduction  from $\overline{\ASAGAP}$ to $\fe{\KK^{\rightarrow}_0}$.

Let $\langle G,s,t\rangle$ be an instance of \ASAGAP, where $G=(V,E)$ is a slice graph with $m$ slices.
We construct the model $\Model{M}_G:=(V,E\cup\{(t,t)\},\xi)$ 
and the formula $\varphi_G:=\Diamond\Box\Diamond\cdots\Box\Diamond(\Diamond\top)$ 
that consists of a sequence of $m-1$ alternating modal operators starting with $\Diamond$
that is followed by $\Diamond\top$.
Notice that $t$ is the only state in $V_m$ 
that has a successor, and therefore it is the only state in $V_m$ where $\Diamond\top$ is satisfied.
Intuitively speaking, the prefix of $\Diamond\top$ in $\varphi_G$ that consists of alternating modal operators 
simulates the alternating path through $G$ from $s$,
and eventually $\Diamond\top$ is satisfied on all the endpoints of this alternating path only if all endpoints equal $t$. 
It is not hard to see that 
an alternating path from $s$ to $t$ exists in $G$ if and only if $\Model{M}_G,s \mmodels \varphi_G$,
i.e. $\langle G,s,t \rangle\in\ASAGAP$ if and only if $\Model{M}_G,s \mmodels \varphi_G$.
Accordingly, $\langle G,s,t \rangle\in\overline{\ASAGAP}$ if and only if $\Model{M}_G,s \mmodels \neg\varphi_G$,
where $\overline{\ASAGAP}$ denotes the complement of $\ASAGAP$.

\markit{We now transform $\neg\varphi_G$ into an equivalent formula in the strictly implicational fragment.
Using duality of $\Diamond$ and $\Box$ we obtain that $\neg\Diamond\Box\Diamond\cdots\Box\Diamond(\Diamond\top)$
is equivalent to $\Box\neg\Box\neg\Box\cdots\neg\Box\neg\Box(\Box\bot)$.
Every subformula $\Box\neg\alpha$ is equivalent to $\Box(\alpha\rightarrow\bot)$,
and the final $\Box(\Box\bot)$ is equivalent to $\Box(\top\rightarrow\Box(\top\rightarrow\bot))$,
where $\top\equiv\Box(\bot\rightarrow\bot)$.
In this way, $\neg\varphi_G$ can be transformed into the equivalent formula $\varphi'_G$
that belongs to the strictly implicational fragment.
It is straightforward that the mapping $\langle G,s,t \rangle\mapsto\langle \varphi'_G,\Model{M}_G,s \rangle$
can be computed in logarithmic space.
Since $\varphi'_G$ contains no variables and belongs to the strictly implicational fragment, 
this yields $\overline{\ASAGAP} \mlogred \fe{\KK^{\rightarrow}_0}$,
and the \p-hardness of $\fe{\KK^{\rightarrow}_0}$ follows from
the \p-completeness of $\ASAGAP$ (Lemma~\ref{lemma:ASAGAP-P-complete}) and the closure of \p under complement.}
\qed

In general, the slice graph is transformed into a frame (of a Kripke model)
to be used in an instance of the model checking problem.
Since the semantics of the logics under consideration 
is defined by Kripke models with frames that are transitive (and reflexive),
we need to produce frames that are transitive (and reflexive).
The straightforward way would be to take the transitive closure of a slice graph.
But this cannot be computed with the resources that are allowed for our reduction functions,
i.e.~in logarithmic space.
Fortunately, slice graphs can easily be made transitive by adding
all edges that ``jump'' from a node to a node that is at least two slices higher.
Clearly, the resulting graph % is not anymore a slice graph, but it
is a transitive supergraph of the transitive closure of the slice graph.
In order to make the reductions from \ASAGAP to the model checking problems work,
the valuation function of the Kripke model and the formula that has to be evaluated
have to be constructed in a way that ``ignores'' these edges that jump over a slice.

\begin{figure}[t]
\mbox{}\hfill
\begin{tikzpicture}[scale=0.6] \ptgraph \end{tikzpicture}
\hfill
\begin{tikzpicture}[scale=0.6] \ptclosuregraph \end{tikzpicture}
\hfill\mbox{}
\caption{A slice graph and its pseudo-transitive closure}
\label{fig:pseudo-trans}
\end{figure}
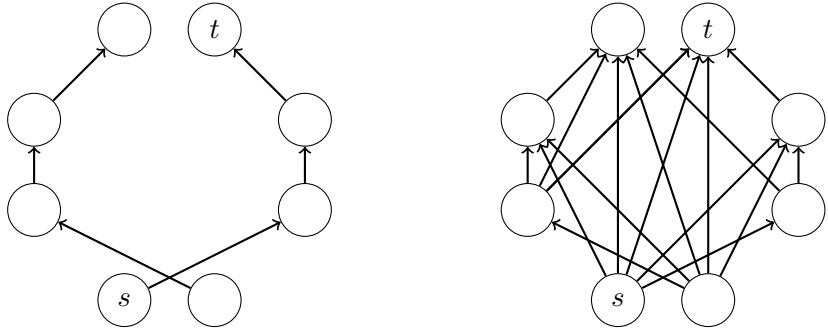

\begin{defi} \label{def:pseudo_trans}
Let $V_{\geq i} = \bigcup_{j=i,i+1,\ldots,m} V_j$, and $V_{\leq i} = \bigcup_{j=1,2,\ldots,i} V_j$.
The \emph{pseudo-transitive closure} of a slice graph $G=(V, E)$
with $m$ slices $V=V_1\cup \cdots \cup V_m$
is the graph $G'=(V,E')$ 
where
 
\begin{mathe}
	E'& := & E  ~~\cup~~ \bigcup_{i=1,2,\ldots,m-2} \klauf V_i \times V_{\geq i+2} \klzu ~~.
\end{mathe}
The \emph{reflexive and pseudo-transitive closure} of the slice graph $G$
is the graph $G''=(V,E'')$ 
where 

\begin{mathe}
	E'' & := & E'  ~~\cup~~ \{(u,u)\mid u\in V\} ~~.
\end{mathe}
\end{defi}\vspace{-6 pt}

\noindent An example for a slice graph and its pseudo-transitive closure
is shown in Figure~\ref{fig:pseudo-trans}.

%% file: 3-intuitionistic_lower_bounds.tex
%%%%%%%%%%%%%%%%%%%%%%%%%%%%%%%%%%%%%%%%%%%%%%%%%%%%%%%%%%%%%%%%%%%%%%%%%%%%%%%%%%%%%%%%%%%%%%%%%%%%%%%%%%%%
% Lowerbounds for intuitionistic logics
%%%%%%%%%%%%%%%%%%%%%%%%%%%%%%%%%%%%%%%%%%%%%%%%%%%%%%%%%%%%%%%%%%%%%%%%%%%%%%%%%%%%%%%%%%%%%%%%%%%%%%%%%%%%

\section{Lower bounds for intuitionistic logics}
\label{sec:intuitionistic lower bounds}

We investigate the complexity of the model checking problem for 
fragments of the intuitionistic logics \KC, \IPC, \FPL, and \BPL in Section~\ref{subsec:p-hard}.
Our basic proof idea is presented in the proof of Theorem~\ref{thm:KC-P-hard}
where we show the \p-hardness of $\fe{\KC^{\iImpl}}$.
This hardness result carries directly over to $\fe{\IPC^{\iImpl}}$ and $\fe{\BPL^{\iImpl}}$.
In order to obtain results for fragments with a restricted number of variables
we extend the construction from the basic proof.
In a first step, we show the \p-hardness of model checking for $\FPL^{\iImpl}$ even if we consider formulas with only one variable, i.e. $\fe{\FPL_1^{\iImpl}}$.
The same proof works for the \p-hardness of $\fe{\BPL_1^{\iImpl}}$.
In a second step, we yield \p-hardness of $\fe{\BPL_0}$.
Notice that it remains open whether $\fe{\BPL_0^{\iImpl}}$ is \p-hard, too.
Our last \p-hardness result in Section~\ref{subsec:p-hard} shows that $\fe{\KC_2}$ and $\fe{\IPC_2}$ are \p-hard.
%This follows almost directly from the proof construction of Theorem 4 in \cite{Rybakov06}.
In Section~\ref{subsec:optimality} we show that the results for 
$\fe{\FPL_1^{\iImpl}}$, $\fe{\KC_2}$, and $\fe{\IPC_2}$ are optimal in the sense, 
that with one variable less the model checking problem cannot be \p-hard, unless unexpected collapses of complexity classes happen.

%First we show \p-hardness for the purely implicational fragments (Section~\ref{subsec:implicational_fragments}).
%These basic proofs then will be extended 
%in order to obtain results for fragments with a restricted number of variables
%(Section~\ref{subsec:bounded_number_of_variables}).
%Finally, we show that all these hardness results are optimal in the sense,
%that using one variable less the model checking problem cannot be \p-hard,
%unless unexpected collapses of complexity classes happen (Section~\ref{subsec:optimality}).

%%------  P-Haerte des implicational fragments fuer FPL

\subsection{\texorpdfstring{\p-hard}{P-hard} fragments.}
\label{subsec:p-hard}
\ \vspace{1ex}

\noindent
We present the basic construction in the proof of Theorem~\ref{thm:KC-P-hard}, where we show the \p-hardness of the model checking problem 
for the implicational fragment of \KC.
For this, we use a logspace reduction from \ASAGAP to $\fe{\KC^{\iImpl}}$.
The \p-hardness of the model checking problems for the implicational fragments of \IPC and \BPL follow straightforwardly.

\input{KCimpl_P-haerte}

Clearly, the same lower bound holds for the implicational fragments of \IPC and \BPL.

\begin{cor}\label{cor:implIPC-BPL-p-hard}
	The model checking problem for $\IPC^{\iImpl}$ and $\BPL^{\iImpl}$ is \p-hard.\qed
\end{cor}

%--------------------------------------------------------------------------------------------------
%  Fragmente mit beschränkter Variablenzahl
%--------------------------------------------------------------------------------------------------

\noindent
The basic construction of the reduction from the above proof can be seen as follows.
The frame of the model contains all information about the \ASAGAP instance
from which it is constructed, but there is some ``noise''
by the pseudo-transitive (and reflexive) edges.
The valuation function gives additional information on the structure of the \ASAGAP instance.
It says where the goal node $t$ sits, 
and it allows to check the distances of any state to the upper most slice.
The formula puts both parts together.
It uses the variables to filter out the original \ASAGAP instance and to evaluate it.

If we restrict the number of variables to be used in the formula,
we need a different approach to measure the distances of the states to the upper most slice.
%
%
%-------- FPL_1 ist P-hart  ------------------------------------------------------------
%
%
For irreflexive frames, we can replace the variables by formulas that measure this distance.
To distinguish the goal node from the other nodes we use one variable.
This yields that $\fe{\FPL_1^{\iImpl}}$ is \p-hard (Theorem~\ref{thm:FPL1-P-hard}).
In Theorem~\ref{thm:FPL0-in-LOGCFL} we show that we cannot save this variable.
Essentially, in the fragment of \FPL without variables we can measure distances, but we cannot do more.

\input{FPL1_P-haerte}

\begin{figure}[t]
\hrulefill
\vspace{1ex}

\fpleinsexample
\caption{The model $\Model{M}$ as constructed from the example instance of \ASAGAP
in Fig.~\ref{fig:kcproofexample} by the proof of Theorem~\ref{thm:FPL1-P-hard}.
For simplicity, the pseudo-transitive edges are not drawn.}
\label{fig:proofexample-FPL1}
\hrulefill
\end{figure}

\begin{cor}\label{cor:implBPL1-p-hard}
	The model checking problem for $\BPL^{\iImpl}_1$ is \p-hard.\qed
\end{cor}

%
%------  BPL_0 ist P-hart  ----------------------------------------------------

For the fragment of \BPL without variables, we can show the \p-hardness 
of model checking only for formulas with the connectives $\iImpl$ and $\vee$.
Our replacement technique for the last variable costs us the implicationality of the fragment.

\input{BPL0_P-haerte}

%--------- IPC_2 und KC_2 sind P-hart  ---------------------------------------------------

%For \IPC with one variable,
%the reflexivity of the states inflates the cost of the yardstick argument
%by making the replacement formulas that read the yardstick exponentially long (see~\cite{MW-STACS11}). 
Other than $\BPL^{\iImpl}_0$ and $\FPL^{\iImpl}_0$,
the implicational fragments of \IPC with any bounded number of variables have only a finite number of equivalence classes (see~\cite{Urq74}).
Therefore they cannot express arbitrary distances in a model.
We obtain \p-hardness of model checking for the fragment of \IPC with two variables,
where the formulas consist of arbitrary connectives. 
The same applies for the fragment of \KC with two variables.

The proof uses our basic construction from the proof of Theorem~\ref{thm:KC-P-hard}
and essentially the same replacement of variables as in the proof of \cite[Theorem 4]{Rybakov06}
showing that the validity problem for $\IPC_2$ is \PSPACE-complete.
Whereas there the reduction works in polynomial-time (that suffices to compute transitive closures),
our construction must be computable in logarithmic space,
and therefore we must deal with the pseudo-transitive closure.
Little other technical changes in the proof are needed.
For completeness, we present the proof in Appendix \ref{appendixA}.

\begin{thm}\label{thm:KC2-P-hard}%\label{cor:IPC2-P-hard}
The model checking problem for $\KC_2$ and for $\IPC_2$ is \p-hard. \qed
\end{thm}

%\input{KC2_P-haerte}
%
%The bound on the number of variables for \KC Theorem~\ref{thm:KC2-P-hard} passes directly to \IPC.
%
%\begin{cor}\label{cor:IPC2-P-hard}
%The model checking problem for $\IPC_2$ is \p-hard.
%\end{cor}
%

%------------------------------------------------------------------------------------------------
%  Optimalität der Variablenzahl
%------------------------------------------------------------------------------------------------

\subsection{Optimality of the bounds of the numbers of variables}
\label{subsec:optimality}
\ \ \vspace{1ex}

\noindent
The \p-hardness of $\fe{\KC_2}$ and $\fe{\IPC_2}$ (Theorem~\ref{thm:KC2-P-hard}) is optimal
because $\fe{\KC_1} \in \NCi$  and $\fe{\IPC_1} \in \ACi$ \cite{MW-STACS11}.
In order to show the optimality of the \p-hardness of $\fe{\FPL^{\iImpl}_1}$ (Theorem \ref{thm:FPL1-P-hard}),
we show that the complexity of $\fe{\FPL_0}$ is below \p.
%\LOGDCFL^{\NLog[1]} \cap \NLog^{\LOGDCFL[1]}

\input{FPL0-upperbound}

It is not known whether $\fe{\FPL_0}$ is \LOGCFL-hard, too.
We show \NLog as lower bound, even for the implicational fragment.

\begin{lem}\label{lem:FPL_0-NLhard}
The model checking problem for $\FPL_0^{\iImpl}$ is \NLog-hard.
\end{lem}

\noindent
\emph{Proof sketch.}
Claim~\ref{longest-path-property} % in the above proof 
shows that in $\FPL_0$ only the depth of a model can be evaluated by a formula.
Accordingly, the $\alpha_i$ formulas can be used to describe the maximal length of a path through a model.
This yields a reduction from the longest path problem in acyclic directed graphs to $\fe{\FPL^{\iImpl}_0}$.
\markit{Let $\langle G=(V,E),v\in V, n\in \mathbb{N} \rangle$ be an instance of the longest path problem.
Then it holds, that the longest path starting in $v$ has the length $n$ if and only if $G,v \imodels \alpha_{i+1}$ and $G,v \not\imodels \alpha_i$.
This follows from Claim~\ref{longest-path-property}(1).
Since \NLog is closed under complementation this is a correct reduction.}
For the \NLog-completeness of this longest path problem see \cite{JT07}.
\qed

%(For a given instance of \GAP one can build a slice graph and use from the second last to the last slide only the edges to the goal node.
%Then ask whether the number of slices is the length of the longest path from the start node.
%So this is the main idea of a reduction from \GAP to this longest path problem for the \NLog-hardness.)

%% file: KCimpl_P-haerte.tex
\begin{thm}\label{thm:KC-P-hard}
The model checking problem for $\KC^{\iImpl}$ is \p-hard.
\end{thm}

\proof
We show $\ASAGAP \mlogred \fe{\KC^{\iImpl}}$.
The result then follows from Lemma~\ref{lemma:ASAGAP-P-complete}.

Let $\langle G,s,t\rangle$ be an instance of \ASAGAP.
We show how to construct a model $\Model{M}_G$ and a formula $\psi_G$ 
such that $\langle G,s,t\rangle\in\ASAGAP$ if and only if $\Model{M}_G,s\imodels \psi_G$.
Let the slice graph $G=(\mathcal{V}, E)$ have $m$ slices, with $\mathcal{V}=V_{\exists}\cup V_{\forall}$,
and $V_{\exists} = V_1\cup V_3\cup \cdots \cup V_{m-1}$, and $V_{\forall} = V_2\cup V_4\cup \cdots \cup V_{m}$.
We use $V_{\geq i}$ to denote $\bigcup_{j\geq i} V_j$.

In order to use $G$ as a frame of a model for \KC, it must be a directed preorder.
To get $(V,\leqslant)$ we build the pseudo-transitive closure of $G$, add the slice $V_{m+1}:=\{\topnode\}$, add edges from every node in $\mathcal{V}$ to $\topnode$, and build the reflexive closure.
It is clear that $(V,\leqslant)$ can be computed from $G$ in logarithmic space.
For simplicity of notation we write
$x<y$ or $y>x$ for $x\leqslant y$ and $x\not=y$, and we also use $x\geqslant y$ and $x>y$ in the same way.
The variables that we will use in our formulas are $a_1,\ldots,a_{m+1}$.
Informally, $a_i$ is satisfied in the states of the slices $V_{i+1},\ldots,V_{m+1}$, further $a_{m}$ is satisfied in the goal node $t$, and $a_{m+1}$ is satisfied in $\topnode$.
Define the valuation function $\xi$ by
$\xi(a_i):= V_{i+1}\cup\cdots\cup V_{m+1}$ (for $i=1,2,\ldots,m-1$), $\xi(a_{m}) := \{t,\topnode\}$, and $\xi(a_{m+1}) := \{\topnode\}$.
The Kripke model $\Model{M}_G = (V, \leqslant, \xi)$ is a model that satisfies the requirements for \KC. 
\begin{figure}[h]
\hrulefill
\vspace{1ex}

\slicegraphexample %
\hfill
\kcproofexample
\caption{A slice graph $G$, and the model $\Model{M}_G$ as constructed in the proof of Theorem~\ref{thm:KC-P-hard}.
The edges to the top node are not drawn,
reflexive edges are not drawn,
and the pseudo-transitive edges are drawn dashed.
The valuation marks the nodes (resp. the slices).
The fat edges indicate that $\apath_G(s,t)$ holds.}

\hrulefill
\label{fig:kcproofexample}
\end{figure}

Figure~\ref{fig:kcproofexample} shows a slice graph $G$ with $m=4$ slices and the Kripke model $\Model{M}_G$ that is transformed from it.
We will use the formulas $\psi_1,\ldots,\psi_m$ in order to express the $\apath_G$ property on $\Model{M}_G$.

\begin{mathe}
	\psi_{m} & := & a_{m}\iImpl a_{m+1} \\
	\psi_{j} & := & \psi_{j+1}\iImpl a_{j} \text{ for } j=m-1,m-2,\ldots,1
\end{mathe}
Next we will show that satisfaction of $\psi_i$ in slice $V_i$ depends only on the edges of the graph $G$ and not on the reflexive and pseudo-transitive edges that were added in order to obtain the Kripke structure.

\begin{claim}\label{claim:phi-property}\label{claim:phi_i-und-phi_i+1}\label{claim:phi-monoton}
For all $i=1,2,\ldots,m-1$ the following holds.
\begin{enumerate}[leftmargin=\parindent]
\item \label{ai-monoton}
      For all $w\in V_{\geq i+1}$ holds $\Model{M}_G,w\imodels \psi_i$.
\item \label{mutual-complement}
      For all $w\in V_i$ holds   
           $\Model{M}_G,w\imodels \psi_{i}$ if and only if $\Model{M}_G,w \not\imodels \psi_{i+1}$.
\item \label{strict-alternation}
For all $w\in V_i$ holds  $\Model{M}_G,w \imodels \psi_{i}$ if and only if $\exists u> w, u\in V_{i+1}: \Model{M}_G,u\not\imodels \psi_{i+1}$.
\end{enumerate}
\end{claim}

\noindent
\emph{Proof of Claim.}
For part (\ref{ai-monoton}), notice that 
$\psi_i=(\cdots((a_{m}\iImpl a_{m+1})\iImpl a_{m-1})\iImpl \cdots \iImpl a_{i+1})\iImpl a_{i}$.
Since $\xi(a_i)=V_{\geq i+1}$, the right-hand side of $\psi_i$ is satisfied in all states in $V_{>i}$.
Therefore $\psi_i$ is satisfied in all states in $V_{\geq i}$, too.

Part (\ref{mutual-complement}) expresses that $\psi_{i}$ and $\psi_{i+1}$ 
behave like the mutual complement in slice $V_i$, and is shown as follows.
Let $w\in V_i$.

\begin{mathe}
	& \hspace{-6.5ex} \Model{M}_G,w \imodels \psi_{i} \\
	\Leftrightarrow & \forall v \geqslant w : \text{if } \mKlasse{M}_G,v \imodels \psi_{i+1} \text{ then } \mKlasse{M}_G,v \imodels a_{i} & 
	                       \text{(semantics of $\iImpl$)} \\
	\Leftrightarrow & \text{if } \mKlasse{M}_G,w \imodels \psi_{i+1} \text{ then } \mKlasse{M}_G,w \imodels a_{i} & \text{(since $\xi(a_i)=V_{>i}$)} \\
	\Leftrightarrow & \mKlasse{M}_G,w \not\imodels \psi_{i+1} & \text{(since $\mKlasse{M}_G,w \not \imodels a_{i}$)}
\end{mathe}

Part (\ref{strict-alternation}) can be proven by proving
$\Model{M}_G,w \not\imodels \psi_{i+1}$ if and only if $\exists u> w, u\in V_{i+1}: \Model{M}_G,u\not\imodels \psi_{i+1}$,
according to (\ref{mutual-complement}).
The direction from right to left follows immediately from part (1) and the monotonicity of intuitionistic logic.
For the other direction,
assume $\forall u> w, u\in V_{i+1}: \Model{M}_G,u\imodels \psi_{i+1}$.
Firstly, this yields $\forall u> w:$ if $\Model{M}_G,u\imodels \psi_{i+2}$ then $\Model{M}_G,u\imodels a_{i+1}$ ($\ast\ast$),
and secondly $\forall u> w, u\in V_{i+1}: \Model{M}_G,u\not\imodels \psi_{i+2}$ (by (\ref{mutual-complement})).
From the latter, it follows
by the monotonicity property of intuitionistic logic that 
$\Model{M}_G,w \not\imodels \psi_{i+2}$. 
Notice that $\Model{M}_G,w \not\imodels a_{i+1}$ by construction of $\xi$,
and therefore we have: if $\Model{M}_G,w \imodels \psi_{i+2}$ then $\Model{M}_G,w \imodels a_{i+1}$.
Together with $(\ast\ast)$ follows 
$\forall u\geqslant w:$ if  $\Model{M}_G,u\imodels \psi_{i+2}$ then $\Model{M}_G,u\imodels a_{i+1}$.
This means $\Model{M}_G,w\imodels\psi_{i+1}$. ~\qedclaim

It is our goal to show that $\psi_1$ is satisfied in state $s\in V_1$ 
if and only if graph $G$ has an alternating $s$-$t$-path, i.e. $\apath_G(s,t)$.
We do this stepwise.

\begin{claim}\label{claim:ipc-alternation}
For all $i=1,2,\ldots,m$ and all $w\in V_i$ holds: 
\begin{enumerate}[leftmargin=\parindent]
\item if $i$ is odd:  $\apath_G(w,t)$ if and only if $\Model{M}_G,w    \imodels \psi_{i}$, and
\item if $i$ is even: $\apath_G(w,t)$ if and only if $\Model{M}_G,w\not\imodels \psi_{i}$.
\end{enumerate}
\end{claim}

\noindent
\emph{Proof of Claim.}
We prove the claim by induction on $i$.
The base case $i=m$ considers an even $i$.
Let $w\in V_m$.
The following equivalences are straightforward.

\begin{mathe}
									& \hspace{-6.5ex} \apath_G(w,t) & \\
	\Leftrightarrow & w=t & \\
%	\Leftrightarrow & \Model{M}_G,w\imodels a_{m} & (\xi(a_m)\cap V_m = \{t\})\\ 
	\Leftrightarrow & \Model{M}_G,w \not\imodels a_{m}\iImpl a_{m+1} ~~~(=\psi_m) & \\
\end{mathe}
For the induction step, consider $i<m$.
First, assume that $i$ is odd.
Then the slice $V_i$ consists of existential nodes.
Let $w\in V_i$.

\begin{mathe}
									& \hspace{-6.5ex} \apath_G(w,t) &  \\
	\Leftrightarrow & \exists u, (w,u)\in E : \apath_G(u,t) & \text{(definition of $\apath_G$)} \\
	\Leftrightarrow & \exists u>w, u\in V_{i+1}:\Model{M}_G,u\not\imodels \psi_{i+1} & \text{(induction hypothesis, construction of $\Model{M}_G$)} \\
	\Leftrightarrow & \Model{M}_G,w \imodels \psi_{i} & \text{(Claim~\ref{claim:phi-property}(\ref{strict-alternation}))}
\end{mathe}
Second, assume that $i$ is even.
Then the slice $V_i$ consists of universal nodes.
Let $w\in V_i$.

\begin{mathe} 
									& \hspace{-6.5ex} \apath_G(w,t) & \\
	\Leftrightarrow & \forall u, (w,u)\in E: \apath_G(u,t) & \text{(definition of $\apath_G$)}\\
	\Leftrightarrow & \forall u>w, u\in V_{i+1}:\Model{M}_G,u\imodels \psi_{i+1} & \text{(induction hypothesis, construction of $\Model{M}_G$)}\\
	\Leftrightarrow & \Model{M}_G,w\not\imodels \psi_{i} & \text{(Claim~\ref{claim:phi-property}(\ref{strict-alternation}))} ~\qedclaim
\end{mathe}

Let $\psi_G:=\psi_1$.
From Claim~\ref{claim:ipc-alternation} it now follows
that $\langle G,s,t\rangle\in \ASAGAP$ if and only if
$\Model{M}_G,s \imodels \psi_{G}$, 
i.e. $\langle \psi_G,\Model{M}_G,s\rangle \in \fe{\KC^{\iImpl}}$.
%The reduction is the mapping
%
%\begin{mathe}
%  \langle G,s,t \rangle & \longmapsto & \langle \psi_G, \Model{M}_G, s \rangle
%\end{mathe}
%
%where $\langle G,s,t \rangle$  is an instance of \ASAGAP
%and $\langle \psi_G, \Model{M}_G, s \rangle$ is an instance of $\fe{\KC^{\rightarrow}}$.
Since $\Model{M}_G$ and $\psi_G$ can be constructed from $G$ using logarithmic space,
it follows that $\ASAGAP \mlogred \fe{\KC^{\iImpl}}$.
\qed

%% file: FPL1_P-haerte.tex
\begin{thm}\label{thm:FPL1-P-hard}
%\fe{\FPL_1} is \p-hard, even for formulas that have $\rightarrow$ as only connective.
The model checking problem for $\FPL_1^{\iImpl}$ is \p-hard.
\end{thm}

\proof

We show $\overline{\ASAGAP} \mlogred \fe{\FPL^{\iImpl}_1}$, where $\overline{\ASAGAP}$ is the complement of $\ASAGAP$.
Since \p is closed under complement, from  Lemma~\ref{lemma:ASAGAP-P-complete} follows that $\overline{\ASAGAP}$ is \p-complete.
Therefore we obtain the \p-hardness of \fe{\FPL^{\iImpl}_1}.

Let $\langle G,s,t \rangle$ with $G=(V,E)$ be an instance of \ASAGAP with $m$ slices.
From that we construct an \fe{\FPL^{\iImpl}_1} instance $\langle \psi, \Model{M}, s \rangle$.
Let $p$ be the variable that is used in $\FPL^{\iImpl}_1$.
Let $(V,\prec)$ be the pseudo-transitive closure of $G$ (see Definition~\ref{def:pseudo_trans}).
We define $\Model{M} := (V,\prec,\xi)$ with $\xi(p):= \{t\}$.
We use $p$ to distinguish $t$ from the other states in slice $V_m$.
Figure~\ref{fig:proofexample-FPL1} shows an example of $\Model{M}$ with $m=4$.

To express the $\apath_G$ property we use the formulas $\psi_m,\psi_{m-1},\dots,\psi_1$ defined as follows.

\renewcommand{\arraystretch}{1.5}\vspace{1.2ex}%
\begin{tabular}{lcllcl}
	$\alpha_m$ & $:=$ & $\bot$, & \hspace{4ex}$\psi_m$ & $:=$ & $p$ \\
	$\alpha_{i}$ & $:=$ & $\top \iImpl \alpha_{i+1}$, & \hspace{4ex}$\psi_{i}$ & $:=$ & $\psi_{i+1} \iImpl \alpha_{i+1}$ \hspace{3ex} for $i=m-1,m-2,\dots,1$
\end{tabular}%
\vspace{1.2ex}\newline\renewcommand{\arraystretch}{0}
Note that the length of $\psi_1$ is approximately the sum of the lengths of all $\alpha_i$ with $m\geq i > 1$, hence it is about $m^2$.
We use the $\alpha_i$ formulas as yardsticks for the slices and the $\psi_i$ formulas for the alternation as we did in the proof of Theorem~\ref{thm:KC-P-hard}.
According to Claim~\ref{claim:phi-property} we give the following claim.
Because of the irreflexivity of $\Model{M}$ we do not need the mutual complement property (Claim~\ref{claim:phi-property}(2)).

\begin{claim}\label{claim:FPLpsi-property}
For all $i$ with $m\geq i \geq 2$ it holds that
\begin{enumerate}[leftmargin=\parindent]
	\item $\Model{M},w \imodels \alpha_i$ if and only if $w \in V_{\geq i+1}$, and
	\item for all $w\in V_{i-1}$ it holds that $\Model{M},w \not\imodels \psi_{i-1}$ if and only if $\exists v\in V_{i},w\prec v : \Model{M},v \imodels \psi_{i}$.
\end{enumerate}
\end{claim}

\noindent
\emph{Proof of Claim.}
With induction on $i$ we show (1).
For $i=m$ it is trivial because $\alpha_m=\bot$.
For the induction step let $w\in W$ and $m > i \geq 2$.

\begin{mathe}
									& \hspace{-6.5ex} \Model{M},w \imodels \alpha_i \hspace{2ex} (=\top \iImpl \alpha_{i+1}) & \\
	\Leftrightarrow & \forall v\in V, w\prec v : \Model{M},v \imodels \alpha_{i+1} & (\text{semantics of } \iImpl) \\
	\Leftrightarrow & \forall v\in V, w\prec v : v \in V_{\geq i+2} & (\text{induction hypothesis}) \\
	\Leftrightarrow & w \in V_{\geq i+1} & (\text{construction of }\Model{M})
\end{mathe}
For (2) consider $w\in V_{i-1}$ with $m \geq i \geq 2$.

\begin{mathe}
									& \hspace{-6.5ex} \Model{M},w \not\imodels \psi_{i-1} \hspace{2ex} (=\psi_i \iImpl \alpha_i)& \\
	\Leftrightarrow & \exists v\in V, w\prec v : \Model{M},v \imodels \psi_i \text{ and } \Model{M},v \not\imodels \alpha_i & (\text{semantics of } \iImpl) \\
	\Leftrightarrow & \exists v\in V_{i}, w\prec v : \Model{M},v \imodels \psi_i & (\text{Claim~\ref{claim:FPLpsi-property}(1)}) ~\qedclaim
\end{mathe}
According to Claim~\ref{claim:ipc-alternation} we have a similar connection between $\apath_G$ and the $\psi_i$ formulas.

\begin{claim}\label{claim:FPL-alternation}
For all $i=m,m-1,\dots,1$ and all $w\in V_i$ it holds that: 
\begin{enumerate}
\item if $i$ is even: $\apath_G(w,t)$ if and only if $\Model{M},w    \imodels \psi_{i}$, and
\item if $i$ is odd:  $\apath_G(w,t)$ if and only if $\Model{M},w\not\imodels \psi_{i}$.
\end{enumerate}
\end{claim}

\noindent
\emph{Proof of Claim.}
We prove this claim by induction on $i$.
The base case $i=m$ considers an even $i$.
Let $w\in V_m$.
The following equivalences are straightforward.

\begin{mathe}
									& \hspace{-6.5ex} \apath_G(w,t) & \\
	\Leftrightarrow & w=t & \\
	\Leftrightarrow & \Model{M},w\imodels p & 
\end{mathe}
The induction step is with the help of Claim~\ref{claim:FPLpsi-property} similar to the induction step in the proof of Claim~\ref{claim:ipc-alternation}.
(Note that the roles of the even and odd slices are swapped.)
We consider $i<m$.
First, assume that $i$ is even.
Then the slice $V_i$ consists of universal nodes.
Let $w\in V_i$.

\begin{mathe}
	& \hspace{-6.5ex} \apath_G(w,t) & \\
	\Leftrightarrow & \forall v \in V, (w,v) \in E : apath_G(v,t) & (\text{definition of } \apath_G) \\
	\Leftrightarrow & \forall v \in V_{i+1}, w\prec v : \Model{M},v \not\imodels \psi_{i+1} & (\text{induction hypothesis, construction of } \Model{M}) \\
	\Leftrightarrow & \Model{M},w \imodels \psi_i & (\text{Claim~\ref{claim:FPLpsi-property}(2)})
\end{mathe}
Second, assume that $i$ is odd, then the slice $V_i$ consists of existential nodes.
Let $w\in V_i$.

\begin{mathe}
	& \hspace{-6.5ex} \apath_G(w,t) & \\
	\Leftrightarrow & \exists v \in V, (w,v) \in E : apath_G(v,t) & (\text{definition of } \apath_G) \\
	\Leftrightarrow & \exists v \in V_{i+1}, w\prec v : \Model{M},v \imodels \psi_{i+1} & (\text{induction hypothesis, construction of } \Model{M}) \\
	\Leftrightarrow & \Model{M},w \not\imodels \psi_i & (\text{Claim~\ref{claim:FPLpsi-property}(2)}) ~\qedclaim
\end{mathe}

Let $\psi:=\psi_1$.
It follows from Claim~\ref{claim:FPL-alternation} that $\Model{M},s \imodels \psi$ (resp. $\langle \psi, \Model{M}, s \rangle \in \fe{\FPL^{\iImpl}_1}$) if and only if $\langle G,s,t \rangle \notin \ASAGAP$.
Since $\Model{M}$ and $\psi$ can be constructed from $G$ using logarithmic space, it follows that $\overline{\ASAGAP} \mlogred \fe{\FPL^{\iImpl}_1}$.
\qed

%% file: BPL0_P-haerte.tex
\begin{thm}\label{thm:BPL0-P-hard}
%\fe{\BPL_0}---\,i.e. the model checking problem for $\BPL$ even without any variables\,---\,is \p-hard.
The model checking problem for $\BPL_0$ is \p-hard.
\end{thm}

\proof
As in the proof of Theorem~\ref{thm:FPL1-P-hard} we show $\overline{\ASAGAP} \mlogred \fe{\BPL_0}$.
%In this proof use the main ideas from the proof of Theorem~\ref{thm:FPL1-P-hard} and combine them with a technique from Rybakov (see~\cite{Rybakov06}).
The proof consists of two parts.
In the first part we modify the construction that we gave in the proof of Theorem~\ref{thm:FPL1-P-hard} in a way that the $\psi_i$ formulas contain two variables but no $\bot$ because we need \markit{$\bot$-free formulas} for the second step.
In the second step we use a technique from Rybakov \cite[Lemma 8]{Rybakov06} to substitute the variables.

Let $\langle G,s,t \rangle$ with $G=(V,E)$ be an instance from \ASAGAP, $(V,\prec)$ be the pseudo-transitive closure of $G$, and $\Model{M}:=(V,\prec,\xi)$ with $\xi(p_1):=\{t\}$ and $\xi(p_2):=\emptyset$.
Informally, $p_2$ plays the role of $\bot$ because for all $w\in V$ it holds that $\Model{M},w \not\imodels p_2$.
We define the $\psi_i$ formulas as mentioned above.

\renewcommand{\arraystretch}{1.5}\vspace{1.2ex}%
\begin{tabular}{lcllcl}
	$\theta_m$ & $:=$ & $p_2$, & \hspace{4ex}$\psi_m$ & $:=$ & $p_1$ \\
	$\theta_{i}$ & $:=$ & $\top \iImpl \theta_{i+1}$, & \hspace{4ex}$\psi_{i}$ & $:=$ & $\psi_{i+1} \iImpl \theta_{i+1}$ \hspace{3ex} for $i=m-1,m-2\dots,1$
\end{tabular}%
\vspace{1.2ex}\newline\renewcommand{\arraystretch}{0}
For the same reason as in the proof of Theorem~\ref{thm:FPL1-P-hard} it holds that

\begin{mathe}
	\Model{M},s \imodels \psi_1 & \Leftrightarrow & \langle G,s,t \rangle \notin \ASAGAP.
\end{mathe}
%
%We can use a technique from Rybakov (see the proof of Lemma 8 in~\cite{Rybakov06}) because $\psi_1$ is positive.
The models $\mathfrak{F}_i=(W_i,R_i)$ for $i=1,2,3$ and the formulas $\beta_1$ and $\beta_2$ are defined as in the proof of Lemma 8 in \cite{Rybakov06}.
Let for $k=1,2,3$

\begin{mathe}
	W_k & := & \gklauf b_k,a_1^k,a_2^k,\dots,a_{k+2}^k \gklzu, \text{ and} \\
	R_k & := & \gklauf (b_k,b_k),(a_{k+2}^k,b_k) \gklzu \cup \gklauf (a_i^k,a_j^k) \mid 1 \leq j < i \leq k+2 \gklzu.
\end{mathe}
The models are depicted in Figure~\ref{fig:BPL0_P-haerte_Modell}. 
\input{Bilder-BPL0_P-haerte_Modell}
The formulas $\beta_1$ and $\beta_2$ are defined as follows.
We use the abbreviation $\alpha_1 := \top \iImpl \bot$ and $\alpha_{i+1} := \top \iImpl \alpha_i$ for $i \geq 1$.

\begin{mathe}
	\beta_1 & := & (\alpha_3 \iImpl \alpha_2) \iImpl ((\alpha_2 \iImpl \alpha_1) \vee \alpha_3) \\
	\beta_2 & := & (\alpha_4 \iImpl \alpha_3) \iImpl ((\alpha_3 \iImpl \alpha_2) \vee \alpha_4) 
\end{mathe}
We define a $\fe{\BPL_0}$-instance $\langle \psi_{\beta},\mathfrak{F}^*,s \rangle$ with $\mathfrak{F}^*=(W^*,R^*)$.

\begin{mathe}
	W^* & := & V \cup W_1 \cup W_2 \cup W_3 \\
	R^{\xi} & := & \gklauf (w,a^1_3),(v,a^2_4),(v,a^3_5) \mid w \in W \setminus \{t\} ~~,v\in W\gklzu \\
	\multicolumn{3}{l}{$R^*$ is the transitive closure of $\prec \cup ~R_1 \cup R_2 \cup R_3 \cup R^{\xi}$}
\end{mathe}	

\noindent Note that $|W_1 \cup W_2 \cup W_3|=15$ and $\prec$ is already transitive, hence one can compute the transitive closure in logarithmic space.
(We give no valuation function because in $\BPL_0$ models variables are irrelevant.)
The connection between $\beta_1$ and $\beta_2$ and $\mathfrak{F}^*$ is shown and explained in Figure~\ref{fig:BPL0_P-haerte_Modell}.
In the following we substitute the variables in $\psi_1$.

\begin{mathe}
	\psi_{\beta} & := & \psi_1[p_1/\beta_1][p_2/\beta_2]
\end{mathe}	
As Rybakov did in the proof of Lemma 8 in \cite{Rybakov06} one can show by induction on the construction of $\psi$ that

\begin{mathe}
	\mathfrak{F}^*,s \imodels \psi_{\beta} & \Leftrightarrow & \Model{M},s \imodels \psi_1.
\end{mathe}
(Note that Rybakov shows this only for $\bot$-free formulas, hence we cannot use the one variable version of $\psi_1$ from the proof of Theorem~\ref{thm:FPL1-P-hard}.)
It holds that $\langle \psi_{\beta},\mathfrak{F}^*,s \rangle \in \fe{\BPL_0}$ if and only if $\langle G,s,t \rangle \notin \ASAGAP$.
It follows directly from the construction that this is a logspace reduction.
\qed

%% file: Bilder-BPL0_P-haerte_Modell.tex
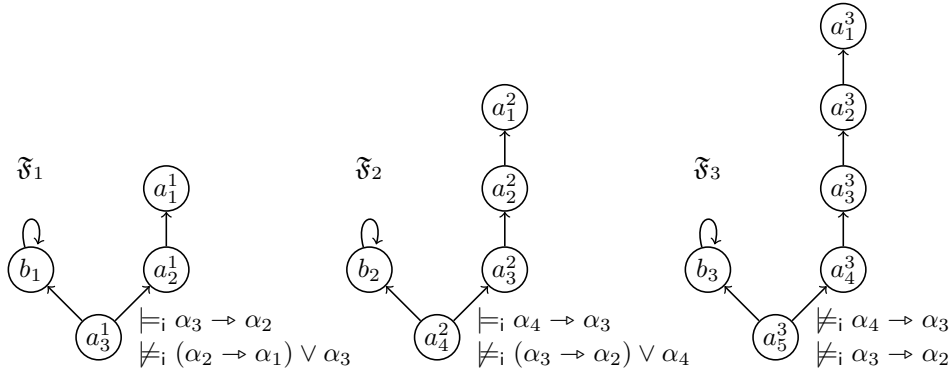
\begin{figure}[t]
\hrulefill
\vspace{1ex}

\begin{tikzpicture}[
 scale=0.9,
 -,
 auto,
 node distance=1.2cm,
 semithick,
 state/.style={style=circle, draw=black, minimum size=6mm, inner sep=0mm},
 empty/.style={style=circle, draw=white, minimum size=0mm, inner sep=0mm},
 txt/.style={style=rectangle},]
 
\node[empty] at(-3,0) {};

%model F_1
\begin{small}
	\node[state] (a31) at (0,0) {$a_3^1$};
	\node[state] (a21) at (1,1) {$a_2^1$};
	\node[state] (a11) at (1,2.2) {$a_1^1$};
	\node[state] (b1) at (-1,1) {$b_1$};
	
	\node[txt,text width=40mm] at (2.8,0.25)  {$\imodels \alpha_3 \iImpl \alpha_2$};
	\node[txt,text width=40mm] at (2.8,-0.25) {$\not\imodels (\alpha_2 \iImpl \alpha_1) \vee \alpha_3$};
	
	\node[txt] at (-1,2.5) {$\mathfrak{F}_1$};
\end{small}

	\path
		(a31) edge[->] (a21)
		(a21) edge[->] (a11)
		(a31) edge[->] (b1)
		(b1) edge[->,loop above] (b1)
	;

%model F_2
\begin{small}
	\node[state] (a42) at (5,0) {$a_4^2$};
	\node[state] (a32) at (6,1) {$a_3^2$};
	\node[state] (a22) at (6,2.2) {$a_2^2$};
	\node[state] (a12) at (6,3.4) {$a_1^2$};
	\node[state] (b2) at (4,1) {$b_2$};

	\node[txt,text width=40mm] at (7.8,0.25)  {$\imodels \alpha_4 \iImpl \alpha_3$};
	\node[txt,text width=40mm] at (7.8,-0.25) {$\not\imodels (\alpha_3 \iImpl \alpha_2) \vee \alpha_4$};
	
	\node[txt] at (4,2.5) {$\mathfrak{F}_2$};
\end{small}

	\path
		(a42) edge[->] (a32)
		(a32) edge[->] (a22)
		(a22) edge[->] (a12)
		(a42) edge[->] (b2)
		(b2) edge[->,loop above] (b2)
	;

%model F_3
\begin{small}
	\node[state] (a53) at (10,0) {$a_5^3$};
	\node[state] (a43) at (11,1) {$a_4^3$};
	\node[state] (a33) at (11,2.2) {$a_3^3$};
	\node[state] (a23) at (11,3.4) {$a_2^3$};
	\node[state] (a13) at (11,4.6) {$a_1^3$};	
	\node[state] (b3) at (9,1) {$b_3$};
	
	\node[txt,text width=40mm] at (12.8,0.25)  {$\not\imodels \alpha_4 \iImpl \alpha_3$};
	\node[txt,text width=40mm] at (12.8,-0.25) {$\not\imodels \alpha_3 \iImpl \alpha_2$};	

	\node[txt] at (9,2.5) {$\mathfrak{F}_3$};
\end{small}

	\path
		(a53) edge[->] (a43)
		(a43) edge[->] (a33)
		(a33) edge[->] (a23)
		(a23) edge[->] (a13)
		(a53) edge[->] (b3)
		(b3) edge[->,loop above] (b3)
	;

\end{tikzpicture}
\caption{
The models $\mathfrak{F}_1$, $\mathfrak{F}_2$, and $\mathfrak{F}_3$. 
(Transitive edges are not depicted.)
It is drawn which parts of $\beta_1$ and $\beta_2$ are satisfied and which are not. 
For example let $w\in W^*$ with $(w,a_3^1) \in R^*$, then it holds that $\mathfrak{F}^*,w \not\imodels \beta_1$. 
Therefore we connect via $R^*$ (resp. $R^{\xi}$) every state from $V \setminus \xi(p)$ with $a_3^1$ and replace $p_1$ with $\beta_1$ in $\psi_1$.
The models $\mathfrak{F}_1$ and $\mathfrak{F}_2$ simulate $\xi$ with technical help of $\mathfrak{F}_3$. 
(For details see the proof of Lemma 8 in \cite{Rybakov06}.)
}
\label{fig:BPL0_P-haerte_Modell}

\hrulefill

\end{figure} 

%% file: FPL0-upperbound.tex
\begin{thm}\label{thm:FPL0-in-LOGCFL}
The model checking problem for $\FPL_0$ is in \LOGCFL.
\end{thm}

\proof
%With respect to $\FPL_0$, every variable free formula 
%can be represented by a small index bounded by the length of the formula \cite{visser80}.
%We will show that every state in an $\FPL_0$ model can also be represented by a small index
%bounded by the number of states in the model.
\markit{Visser \cite{visser80} gives a systematically construction of representatives of the formula equivalence classes of variable free formulas over irreflexive Kripke models.
This enables that every variable free formula can be represented by a small string.
We call this string \textit{formula index}.
We will show that every state in an $\FPL_0$ model can also be represented by the length of its longest outgoing path.
It turns out, that a formula is satisfied in a state if and only if the formula index is greater than the length of the longest path that starts in the state.}
This yields a \LOGCFL algorithm for the model checking problem for $\FPL_0$.

The formula index of a formula is the index $i$ of the $\FPL_0$-equivalent%
\footnote{Two variable free $\Klasse{IL}$ formulas $\varphi$ and $\psi$ are \textit{$\FPL_0$-equivalent} if for all states $w$ in all $\FPL_0$ models $\Model{M}$ it holds that $\Model{M},w \imodels \varphi ~\Leftrightarrow~ \Model{M},w \imodels \psi$. We denote this as $\varphi \equiv_F \psi$.}
formula $\alpha_i$ from \cite[Def. 4.3]{visser80} defined as follows.
Let $i \in \mathbb{N}\cup \{\omega\}$, where $\omega > i$ for all $i \in \mathbb{N}$.

\renewcommand{\arraystretch}{1.2}\vspace{1.7ex}
\begin{tabular}{lcllcllcl}
	$\alpha_0$ & $:=$ & $\bot$, & \hspace{4ex}$\alpha_{\omega}$ & $:=$ & $\top$, & $\hspace{4ex}\alpha_{i+1}$ & $:=$ & $\top \iImpl \alpha_{i}$ \hspace{1.5ex}for $i \in \mathbb{N}$.
\end{tabular}
\vspace{1ex}\renewcommand{\arraystretch}{1}
\begin{claim}\label{claim:V-formula_eq} \cite[Fact 4.4(iii)]{visser80}
Every variable free $\Klasse{IL}$ formula is $\FPL_0$-equivalent to exactly one $\alpha_i$.
\end{claim}

\noindent
One can prove the claim with the following case distinction \cite[Fact 4.4(ii)]{visser80}.

\renewcommand{\arraystretch}{1.2}\vspace{1.7ex}
\begin{tabular}{ll}
	If $\varphi = \bot$, & then \hspace{1.15ex} $\varphi \equiv_F \alpha_0$. \\
%	If $\varphi = \top$, & then $\varphi \equiv_F \alpha_{\omega}$. \\
	If $\varphi \equiv_F  \alpha_a \wedge \alpha_b$, & then \hspace{1.15ex} $\varphi \equiv_F \alpha_{\min\{a,b\}}$. \\
	If $\varphi \equiv_F  \alpha_a \vee \alpha_b$, & then \hspace{1.15ex} $\varphi \equiv_F \alpha_{\max\{a,b\}}$. \\
	If $\varphi \equiv_F  \alpha_a \iImpl \alpha_b$, & then $\begin{cases} \varphi \equiv_F \alpha_{\omega} & \text{~~if } a\leq b\\
																																				\varphi \equiv_F \alpha_{b+1} & \text{~~if } a> b.
																													\end{cases}$
%	If $\varphi \equiv  \alpha_a \rightarrow \alpha_b$ and $a \leq b$, & then $\varphi \equiv \alpha_{\omega}$. \\
%	If $\varphi \equiv  \alpha_a \rightarrow \alpha_b$ and $a > b$, & then $\varphi \equiv \alpha_{b+1}$.
\end{tabular} 
\vspace{1.7ex}\newline\renewcommand{\arraystretch}{1}
If $\varphi\equiv_F \alpha_i$, we call $i$ the \emph{formula index} of $\varphi$.
In order to analyse the complexity of the formula index computation, we define the following decision problem.

\dproblem{\EQVF}
				 {$\langle \varphi,i \rangle$, where $\varphi$ is a variable free $\Klasse{IL}$ formula and $i \in \mathbb{N} \cup \{\omega\}$}
				 {Is $\alpha_i \equiv_F \varphi$?} 

\vspace{-1ex}
\begin{claim}\label{claim:LOGDCFL_eq_visserformel}
\EQVF is in \markit{\LOGCFL}.
%\EQVF is in \LOGDCFL.
\end{claim}

\noindent
\emph{Proof of Claim.}
From the case distinction above one can directly form a recursive algorithm.
If $\varphi \equiv_F \alpha_i$ it holds that $i=\omega$ or $i \leq |\varphi|$. ($|\varphi|$ denotes the length of $\varphi$.)
So every variable value can be stored in logarithmic space.
The algorithm walks recursively through the formula and computes the formula index of every subformula once, hence running time is polynomial.
All information that are necessary for recursion can be stored on the stack.
Therefore the algorithm can be implemented on a polynomial time logspace machine that uses an additional stack i.e. a \LOGCFL-machine \markit{(even without using nondeterminism)}. ~\qedclaim
%
%\begin{algorithm}[h]
%	\caption{Visser formula index check.}
%	\label{alg:eq_V_formula}
%  \begin{algorithmic}[1]
%
%		\REQUIRE an $\Klasse{IL}_0$ formula $\varphi$ and an index $i \in \mathbb{N} \cup \{\omega\}$
%		\STATE \textbf{if} VIndex-calc$(\varphi)=i$ \textbf{then} accept \textbf{else} reject
%			\par\vspace{5pt}
%			
%		\STATE\textbf{function} VIndex-calc($\psi$) \hspace{4ex}// requires an $\Klasse{IL}_0$ formula, returns an index from $\mathbb{N} \cup \{\omega\}$
%		\STATE{\textbf{if} $\psi=\bot$ \textbf{then return} $0$}
%%		\STATE{\textbf{else if} $\psi=\top$ \textbf{then return} $\omega$}
%		
%		\STATE{\textbf{else if} $\psi=\beta \wedge \gamma$ \textbf{then return} $\min \{ \text{VIndex-calc}(\beta), \text{VIndex-calc}(\gamma)\}$}
%		\STATE{\textbf{else if} $\psi=\beta \vee \gamma$ \textbf{then return} $\max \{ \text{VIndex-calc}(\beta), \text{VIndex-calc}(\gamma)\}$}
%		\STATE{\textbf{else if} $\psi=\beta \rightarrow \gamma$ \textbf{then}}
%		\STATE{\hspace{4ex} $b:=\text{VIndex-calc}(\beta)$}
%		\STATE{\hspace{4ex} $c:=\text{VIndex-calc}(\gamma)$}
%		\STATE{\hspace{4ex} \textbf{if} $b \leq c$ \textbf{then return} $\omega$} \textbf{else return} $c+1$
%%		\STATE{\hspace{4ex} \textbf{else return} $c+1$}
%
%		\STATE{\textbf{end if}}
%	\end{algorithmic}
%\end{algorithm}
%

In the following we show that for model checking every $\FPL_0$ model can be reduced to its longest path.
Let $\Model{M}=(W,R)$ be an $\FPL_0$ model.
(Note that we need no valuation function because in $\FPL_0$ models variables are irrelevant.)
Therefore we define a function $\lp_{\Model{M}} : W \rightarrow \mathbb{N}$ that maps a state $w$ to the length of the longest path in $\Model{M}$ starting in $w$.

\begin{mathe}
	\lp_{\Model{M}}(w) & := & \begin{cases}
																\hspace{1.2ex} 0, & \text{~~~if } \nexists v \in W : (w,v) \in R \\
																\max\limits_{(w,v) \in R}\{\lp_{\Model{M}}(v)\}+1, & \text{~~~otherwise}
														\end{cases}
\end{mathe}
\begin{claim}\label{longest-path-property}
\hfill
% Let $\Model{M}=(W,R)$ be an $\FPL_0$ model.
\begin{enumerate}[leftmargin=\parindent]
	\item Let $\Model{M}=(W,R)$ be an $\FPL_0$ model. For every $\alpha_i$ % with $i \in \mathbb{N}\ \cup \{\omega\}$ 
	      and every state $w\in W$ it holds that $\Model{M},w\imodels \alpha_i$ if and only if $\lp_{\Model{M}}(w) < i$.
	\item The following problem is \NLog-complete: given an $\FPL_0$ model $\Model{M}$, an integer $n$, and 
	      a state $w$ of $\Model{M}$; does $\lp_{\Model{M}}(w)=n$ hold?
\end{enumerate}
\end{claim}

\noindent
\emph{Proof of Claim.}
We prove (1) with induction on the formula index $i$.
The cases $i=0$ and $i=\omega$ are clear.
The induction step is shown by the following equivalences.

\begin{mathe}
									& \hspace{-6.5ex} \Model{M},w \imodels \alpha_{i+1} \hspace{2ex} (= \top \iImpl \alpha_i) & \\
	\Leftrightarrow~~ &	\forall s\in W, (w,s) \in R : \Model{M},s \imodels \alpha_i & (\text{semantics of }\iImpl) \\
	\Leftrightarrow~~ &	\forall s\in W, (w,s) \in R : ~~\lp_{\Model{M}}(s) < i& (\text{induction hypothesis}) \\
	\Leftrightarrow~~ &	\lp_{\Model{M}}(w) < i+1 & (\text{irreflexivity of } \Model{M})
\end{mathe}

\noindent For (2) note that the problem for a given graph $G$, a node $s$ of $G$ and an integer $n$ to decide whether the longest path in $G$ starting in $s$ has the length $n$ is \NLog-complete \cite{JT07}. ~\qedclaim

%\begin{algorithm}[h]
%	\caption{Find longest path.}
%	\label{alg:lp_algo}
%  \begin{algorithmic}[1]
%
%		\REQUIRE an $\FPL_0$ model $\Model{M}$, a state $w$ of $\Model{M}$ and an integer $n$
%		\STATE guess nondeterministically a state $v$ 
%		\STATE guess a path from $w$ to $v$ state by state and verify that this path has the length $n$.
%		\STATE \textbf{if} for all states $v$ the length of the path from $w$ to $v$ is at most $n$ \textbf{then} accept \textbf{else} reject
%
%	\end{algorithmic}
%\end{algorithm}

\begin{algorithm}[h]
	\caption{$\FPL_0$ model checking algorithm.}
	\label{alg:FPL0-MC}
  \begin{algorithmic}[1]
		\REQUIRE a variable free $\Klasse{IL}$ formula $\varphi$, an $\FPL_0$ model $\Model{M}$, and a state $w$ from $\Model{M}$
		\STATE guess nondeterministically a formula index $i \in \{0,1,\ldots,|\varphi|\} \cup \{\omega\}$ 
		\STATE \textbf{if} $(\varphi,i) \in \EQVF$ \textbf{then}
%		\STATE use Algorithm \ref{alg:eq_V_formula} to decide whether $(\varphi,i) \in \EQVF$
		\STATE \hspace{4ex} guess nondeterministically an integer $n < i$ 
		\STATE \hspace{4ex} \textbf{if} $\lp_{\Model{M}}(w) = n$ \textbf{then} accept \textbf{else} reject
%		\STATE verify with Algorithm \ref{alg:lp_algo} that $\lp_{\Model{M}}(w) = n$
		\STATE \textbf{else} reject
%		\STATE \textbf{if} $n < i$ \textbf{then} accept \textbf{else} reject
	\end{algorithmic}
\end{algorithm}
Algorithm \ref{alg:FPL0-MC} decides $\fe{\FPL_0}$ with the resources of \LOGCFL.
In the first two steps we compute the formula index of $\varphi$. 
With Claim \ref{claim:LOGDCFL_eq_visserformel} it follows that these steps can be done with the resources of \LOGCFL.
In the next steps the length of the longest path starting in $w$ is guessed and verified.
The verification (Step 4) can be done with the resources of \NLog.
The correctness of Step 4 follows from Claim~\ref{claim:V-formula_eq} and Claim \ref{longest-path-property}.
Altogether Algorithm \ref{alg:FPL0-MC} can be implemented on a nondeterministic polynomial time machine with logarithmic space and an additional stack.
These are the resources of \LOGCFL.
\qed

%% file: 4-modal_lower_bounds.tex
%%%%%%%%%%%%%%%%%%%%%%%%%%%%%%%%%%%%%%%%%%%%%%%%%%%%%%%%%%%%%%%%%%%%%%%%%%%%%%%%%%%%%%%%%%%%%%%%%%%%%%%%%%%%
% Lowerbounds for modal logics
%%%%%%%%%%%%%%%%%%%%%%%%%%%%%%%%%%%%%%%%%%%%%%%%%%%%%%%%%%%%%%%%%%%%%%%%%%%%%%%%%%%%%%%%%%%%%%%%%%%%%%%%%%%%

\section{Lower bounds for modal logics}
\label{sec:modal lower bounds}

For all \p-hard model checking problems for fragments of intuitionistic logics 
we obtain the same lower bound for their modal companions.

\begin{thm}\label{thm:lowerbounds_modalL} 
The model checking problem is \p-hard for $\Kiv_0$, $\PrL^{\rightarrow}_1$, $\Sivii^{\rightarrow}$, $\Kiv_1^{\rightarrow}$, and $\Siv^{\rightarrow}$.
\end{thm}

\begin{proof}
By Lemma~\ref{lem:Goedel-Tarski} this follows from Theorems~\ref{thm:BPL0-P-hard}, \ref{thm:FPL1-P-hard}, and \ref{thm:KC-P-hard}.
\end{proof}

From Theorem~\ref{thm:KC2-P-hard} and Lemma~\ref{lem:Goedel-Tarski} we obtain that 
the model checking problem for $\Sivii_2$---the modal companion of $\KC_2$---is \p-hard.
Even though model checking for $\KC_1$ is in $\NCi$~\cite{MW-STACS11}, 
we can show that one variable suffices to make model checking \p-hard for $\Sivii$.

\input{S421_P-haerte}

Note that the reduction in the proof of Theorem~\ref{thm:S4.2_1-P-hard} is not suitable for intuitionistic logics,
since the constructed Kripke model lacks the monotonicity property of the variables.
Moreover, in that proof we make extensive use of negation, that would have a very different meaning in intuitionistic logics.

Clearly, the same lower bound holds for the fragment of $\Siv$ with one variable.

\begin{cor}\label{cor:S4_1-P-hard}
The model checking problem for $\Siv_1$ is \p-hard. \qed
\end{cor}

The \p-hardness results for $\fe{\Sivii_1}$ and $\fe{\Siv_1}$ are optimal
since the model checking problem for $\Siv_0$ is easy to solve.
A formula without any variables is either satisfied by every model w.r.t.~\Siv or it is satisfied by no model.
This is because $\Diamond \top$ (resp. $\Box \top$) is satisfied by every state in every model, 
and $\Diamond \bot$ (resp. $\Box \bot$) is satisfied by no state in every model.
Essentially, in order to evaluate a $\Siv_0$ formula in some model,
the model and the modal operators can be ignored and the remaining classical propositional formula can be evaluated 
like a classical propositional formula---this problem is \NCi-complete (see \cite{bus87}).

\begin{lem}\label{lem:S40_in_NC1}
The model checking problem for $\Siv_0$ and for $\Sivii_0$ are \NCi-complete. \qed
\end{lem}

According to Theorem~\ref{thm:FPL0-in-LOGCFL} we show that the complexity of $\fe{\PrL_0}$ is below \p, namely $\fe{\PrL_0} \in \ACi$.
Therefore the \p-hardness of $\fe{\PrL^{\rightarrow}_1}$ is optimal in the sense that we cannot save the variable.

\input{PrL0-upperbound}

It is not known whether \ACi also is the lower bound of $\fe{\PrL_0}$.
But from Lemmas~\ref{lem:Goedel-Tarski} and \ref{lem:FPL_0-NLhard}, the lower bound \NLog follows, even for the strictly implicational fragment.

\begin{lem}\label{lem:PrL_0-NLhard}
The model checking problem for $\PrL_0^{\rightarrow}$ is \NLog-hard. \qed
\end{lem}

Even though we do not know the exact complexity of $\fe{\FPL_0}$ and $\fe{\PrL_0}$,
it is a bit surprising that the \LOGCFL upper bound we got for $\fe{\FPL_0}$ (Theorem~\ref{thm:FPL0-in-LOGCFL})
is lower than the \ACi upper bound for $\fe{\PrL_0}$ (Theorem~\ref{thm:PrL0-in-AC1}).

%% file: S421_P-haerte.tex
\begin{thm}\label{thm:S4.2_1-P-hard}
The model checking problem for $\Sivii_1$ is \p-hard.
\end{thm}

\proof
We show that $\ASAGAP \mlogred \fe{\Sivii_1}$.
Since \ASAGAP is \p-hard (Lemma~\ref{lemma:ASAGAP-P-complete}),
the \p-hardness of \fe{\Sivii_1} follows.

Let $\langle G,s,t\rangle$ be an instance of \ASAGAP,
where $G=(V_{\exists}\cup V_{\forall}, E)$ is a slice graph with $m$ slices,
and $V_{\exists} = V_1\cup V_3\cup \cdots \cup V_{m-1}$, 
and $V_{\forall} = V_2\cup V_4\cup \cdots \cup V_{m}$.
We construct a Kripke model $\Model{M}_G = (U, R, \xi)$ and a formula $\lambda_1$
such that $\langle G,s,t \rangle\in\ASAGAP$ if and only if $\langle \lambda_1,\Model{M}_G,s \rangle\in\fe{\Sivii_1}$.
First, let $G_t=(V,\leqslant)$ be the pseudo-transitive and reflexive closure of $G$.
%(as constructed in the proof of Theorem~\ref{thm:IPC-P-hard}).
%
Second, we add two slices to $G_t$, namely $V_{m+1}:=\{u,t_1,t_2\}$ and $V_{m+2}:=\{\topnode\}$.
Third, we add the edges $\{(v,u) \mid v\in V_m\}$ from every node in $V_m$ to $u$,
edges  $\{(t,t_1), (t,t_2)\}$ from the goal node $t\in V_m$ to $t_1$ and to $t_2$,
and edges $\{(u,\topnode),(t_1,\topnode),(t_2,\topnode)\}$ from every node in $V_{m+1}$ to $\topnode$.
Moreover, in slice $V_{m+1}$ we abstain from the rule that there are no edges
between different nodes in the same slice.
We also add the edges $\{(t_1,t_2), (t_2,t_1)\}$ between $t_1$ and $t_2$ in both directions.
Finally, we add pseudo-transitive edges $V_{\leq m-1}\times V_{m+1}$ and $V_{\leq m}\times V_{m+2}$,
and reflexive edges to all nodes.
Let the graph $G'=(U,R)$ be the graph obtained in this way.
Then $G'$ is reflexive, transitive, and every node has an edge to $\topnode$.
Therefore, $G'$ is a directed preorder.

In order to be able to find out in which slice a state is,
we mark every even slice $V_2,V_4,\ldots,V_m,V_{m+2}$ with the variable $a$,
and in slice $V_{m+1}$ the node $t_2$ is marked with $a$.
This yields
the valuation function $\xi$ to be defined by $\xi(a):= V_2 \cup V_4 \cup \cdots \cup V_{m+2}\cup \{t_2\}$,
and completes the construction of the Kripke model $\Model{M}_G:=(U,R,\xi)$.
Figure~\ref{fig:proofexamples42} shows an example.

\begin{figure}[t]
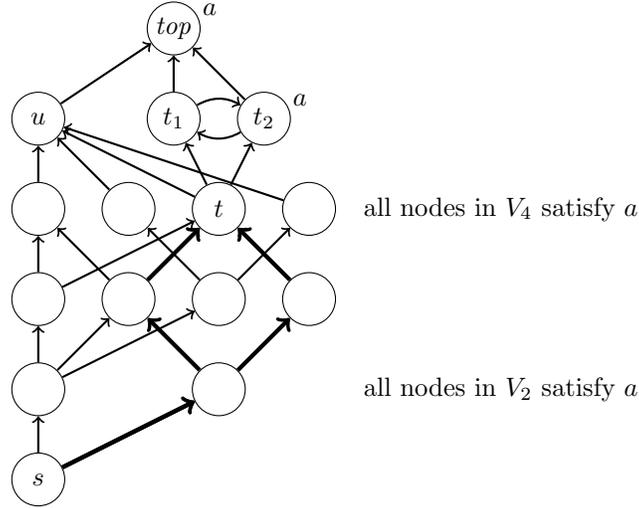

\hrulefill
\vspace{1ex}

\siviiexample %
\caption{The model $\Model{M}_G$ as constructed in the proof of Theorem~\ref{thm:S4.2_1-P-hard}
for the ASAGAP instance from Figure~\ref{fig:kcproofexample}.
Pseudo-transitive edges and reflexive edges are not drawn for simplicity.
The valuation marks the nodes (resp. the slices).
The fat edges indicate that $\apath_G(s,t)$ holds.}
\label{fig:proofexamples42}

\hrulefill
\end{figure}

Let $\eta := \neg a \wedge \Diamond(a \wedge \Diamond \neg a)$.
We will use that $\eta$ is satisfied in $t_1$, but it is not satisfied in $V_m\cup\{u,t_2,\topnode\}$.
The goal node $t$ is the only node in slice $V_m$ 
that has a successor (namely $t_1$), in which $\eta$ is satisfied.
We can estimate the slice to which a node belongs using the following formulas $\delta_i$.
Let $\delta_{m} := \Diamond(\neg \eta)$, and for $i=m-1,m-2,\ldots,1$

\begin{mathe}
	\delta_{i} & := & \begin{cases}
												\Diamond( \neg a \wedge \delta_{i+1}), 		& \text{~~~if~~~} i \text{ is even,} \\
												\Diamond( ~~~~~~a~~ \wedge \delta_{i+1}), & \text{~~~if~~~} i \text{ is odd.}
										\end{cases}
\end{mathe}	
%									
%$$
%\delta_{i} := \left\{
%\begin{array}{rl}
% \Diamond( \neg a \wedge \delta_{i+1}), & \mbox{ if $i$ is even,} \\
% \Diamond( ~~~~~~a~~ \wedge \delta_{i+1}), & \mbox{ if $i$ is odd.}
%\end{array}
%\right.
%$$
%
%Let $V_{\leq i}$ denote $V_1\cup\cdots\cup V_i$.
\vspace{-1ex}
\begin{claim}\label{claim:delta}
Let $x\in V_{\leq m}$ and $i=1,2,\ldots,m$.
Then $\Model{M}_G,x \mmodels \delta_i$ if and only if
$x\in V_{\leq i}$.
\end{claim}

\noindent
\emph{Proof of Claim.}
We proceed by induction on $i=m,m-1,\ldots,1$.
The base case $i=m$ is clear, since $\neg \eta$ is satisfied in $u$
and every state in $V_{\leq m}$ has an edge to $u$.
For the induction step consider an arbitrary $i<m$.
Let $i$ be odd and  and $x\in V_{\leq m}$.
If $\Model{M}_G,x \mmodels \delta_i$, 
then $x$ has a successor $y$ with
$\Model{M}_G,y \mmodels a$ and $\Model{M}_G,y \mmodels \delta_{i+1}$.
By the induction hypothesis we obtain $y\in V_{\leq i+1}$.
If $x\not= y$, it follows by the properties of the slice graph that 
$y$ is a successor of $x$ in a slice ``higher'' than that of $x$.
%, and therefore $x\in V_{\leq i}$.
%If $x=y$, then $y\not\in V_{i+1}$ because $V_{i+1}\subseteq \xi(a)$ and $\Model{M}_G,y \mmodels \neg a$.
The case $x=y$ is not possible because $\Model{M}_G,x \mmodels \neg a$ and $\Model{M}_G,y \mmodels a$.
Therefore $x \in V_{\leq i}$.
For the other proof direction,
take any $x\in V_{\leq i}$.
The formula $\delta_i$ is satisfied in $x$,
if there exists a path of length $m-i+1$ from $x$ to $u$ in $(U,R)$,
that goes through states that alternatingly satisfy $a$ and $\neg a$.
This means, that no edge $(v,v)$ appears on this path.
Since every state in $V_{\leq m}$ has a successor in the subsequent slice,
such a path exists, and therefore $\Model{M}_G,x \mmodels \delta_i$.
For even $i$, the proof is similar. ~\qedclaim

The goal state $t$ is the only state in $V_m$ that satisfies $\Diamond \eta$.
Using the $\delta_i$ formulas to verify an upper bound for the slice of a state,
we can now simulate the alternating graph accessibility problem by the following formulas.

Let $\lambda_{m} := a \wedge \Diamond \eta$ and for $i=m-1,m-2,\ldots,1$

\begin{mathe}
	\lambda_{i} &:= & \begin{cases}
												\neg a \wedge \Diamond (\delta_{i+1} \,\wedge\, \lambda_{i+1}), & \text{~~~if~~~} i \text{ is odd,} \\
												a \wedge \Box (\delta_{i+1} \rightarrow \lambda_{i+1}), 				& \text{~~~if~~~} i \text{ is even.}
										\end{cases}
\end{mathe}										
%
%$$
%\lambda_{i} := \left\{
%\begin{array}{rl}
%\neg a \wedge \Diamond (\delta_{i+1} \,\wedge\, \lambda_{i+1}), & \mbox{ if $i$ is odd,} \\
%a \wedge \Box (\delta_{i+1} \rightarrow \lambda_{i+1}), & \mbox{ if $i$ is even.}
%\end{array}
%\right.
%$$
\vspace{-1ex}
\begin{claim}\label{claim:s42main}
For $i=1,2,\ldots,m$ and all $x\in V_i$ holds:
$\apath_G(x,t)$ if and only if $\Model{M}_G,x \mmodels \lambda_i$.
\end{claim}

\noindent
\emph{Proof of Claim.}
We prove the claim by induction on $i$ 
and start with $i=m$.
For all $x\in V_m$ holds $\Model{M}_G,x \mmodels \lambda_m$ if and only if $x=t$, where the latter is the same as $\apath_G(x,t)$.
For the induction step,
consider an odd $i<m$ first and let $x\in V_i$.
We get the following equivalences.

\begin{mathe}
									& \hspace{-6.5ex} \apath_G(x,t) &  \\
	\Leftrightarrow &	\exists (x,y)\in E: y\in V_{i+1} \text{~and~} \apath_G(y,t) & \text{(definition of $\apath_G$)} \\
	\Leftrightarrow & \exists (x,y) \in R: \Model{M}_G,y\mmodels \delta_{i+1} \text{~and~} \Model{M}_G,y \mmodels \lambda_{i+1} & 
	                     \text{(ind. hypoth., Claim~\ref{claim:delta})} \\
	\Leftrightarrow & \Model{M}_G,x \mmodels \neg a \wedge \Diamond ( \delta_{i+1} \wedge \lambda_{i+1} ) \hspace{2ex}(= \lambda_i) & \text{(construction of $\Model{M}_G$)}
\end{mathe}
%
%From the definition of $\apath_G$ for slice graphs follows (i) iff (ii),
%(ii) iff (iii) by the inductive hypothesis and by Claim~\ref{claim:delta},
%and (iii) iff (iv) by the construction of $\Model{M}_G$.
%
Second, consider an even $i<m$, and let $x\in V_i$.
The following equivalences hold.

\begin{mathe}
									& \hspace{-6.5ex} \apath_G(x,t) &  \\
	\Leftrightarrow & \forall (x,y) \in E\!: \text{~if~} y \in V_{i+1} \text{~then~} \apath_G(y,t) & \\
	\Leftrightarrow & \forall (x,y) \in R\!: \text{~if~} \Model{M}_G,y\mmodels \delta_{i+1} \text{~then~} \Model{M}_G,y \mmodels \lambda_{i+1} &  \\
	\Leftrightarrow & \Model{M}_G,x \mmodels a \wedge \Box ( \delta_{i+1} \rightarrow \lambda_{i+1} ) \hspace{2ex}(= \lambda_i) & 
\end{mathe}
%
%From the definition of $\apath_G$ for slice graphs follow (i) iff (ii),
%(ii) iff (iii) by the inductive hypothesis and by Claim~\ref{claim:delta},
%and (iii) iff (iv) by the construction of $\Model{M}_G$.
The arguments for the equivalences are the same as above. ~\qedclaim

%%%%%%%%%%%%%%%%%%%%%%%%%%%%%%%%%%%%%%%%%%%%%%%%%%%%%%%%%%%%%%%%%

From Claim~\ref{claim:s42main} it now follows
that $\langle G,s,t\rangle\in \ASAGAP$ if and only if
$\Model{M}_G,s \mmodels \lambda_1$,
i.e. $\langle \lambda_1, \Model{M}_G,s\rangle \in \fe{\Sivii_1}$.
Since the construction of $\Model{M}_G$ and $\lambda_1$ from $G$ can
be computed in logarithmic space,
it follows that $\ASAGAP \mlogred \fe{\Sivii_1}$.
\qed

%% file: PrL0-upperbound.tex
\begin{thm}\label{thm:PrL0-in-AC1}
The model checking problem for $\PrL_0$ is in \ACi.
\end{thm}

\proof
We show that every $\PrL_0$ model can be reduced to its longest path.
%\footnote{Let $\Model{M}_1$ and $\Model{M}_2$ be $\PrL_0$ models and $w_1$ (resp. $w_2$) be a state of $\Model{M}_1$ (resp. $\Model{M}_2$).
%We say $(\Model{M}_1,w_1)$ is \textit{homomorphic} to $(\Model{M}_2,w_2)$ if for every $\Klasse{ML}_0$ formula $\varphi$ it holds that $\Model{M}_1,w_1 \mmodels \varphi \Leftrightarrow \Model{M}_2,w_2 \mmodels \varphi$.}
Therefore we define linear models%
\footnote{A frame $\Model{M}=(W,R)$ is \textit{linear} if for every $w_1,w_2 \in W$ (with $w_1 \neq w_2$) it holds that either $(w_1,w_2) \in R$ or $(w_2,w_1) \in R$.}
$\Model{L}_n := (\{0,1,\ldots,n\},>)$ and use the function $\lp_{\Model{M}}$, that maps a state to the length of the longest path in its model starting in this state (see the proof of Theorem~\ref{thm:FPL0-in-LOGCFL}).
(Note that we give no valuation function because in $\PrL_0$ models variables are irrelevant.)
Reinhardt \cite{KR-presonal11} recently showed the upper bound \ACi for $\PrL_0$ model checking restricted to linear models.

%For a given $\PrL_0$ model $\Model{M}$ and a state $w$ we show that $(\Model{M},w)$ is homomorphic to the linear model that has the length of the longest path starting of $\Model{M}$ in $w$.

\begin{claim}\label{claim:homomorphic_L_model}
Let $\Model{M}=(W,R)$ be a $\PrL_0$ model, $w\in W$, and $\varphi$ a variable free $\Klasse{ML}$ formula.
Then it holds that $\Model{M},w \mmodels \varphi$ if and only if $\Model{L}_{\lp_{\Model{M}}(w)},\lp_{\Model{M}}(w) \mmodels \varphi$.
\end{claim}

\noindent
\emph{Proof of Claim.}
We show this by induction on the construction $\varphi$.
The case $\varphi=\bot$ is clear.
In the induction step the case $\varphi=\alpha \rightarrow \beta$ is straightforward.
Assume that $\varphi= \Box \alpha$.

\begin{mathe}
	& \hspace{-6.5ex} \Model{M},w \models \varphi \hspace{2ex} (= \Box \alpha) & \\
	\Leftrightarrow & \forall v \in W, (w,v) \in R : \Model{M},v \mmodels \alpha & (\text{semantics of } \Box) \\
	\Leftrightarrow & \forall v \in W, (w,v) \in R : \Model{L}_{\lp_{\Model{M}}(v)},\lp_{\Model{M}}(v) \mmodels \alpha & (\text{induction hypothesis}) \\
	\Leftrightarrow & \forall v \in W, (w,v) \in R : \Model{L}_{\lp_{\Model{M}}(w)},\lp_{\Model{M}}(v) \mmodels \alpha & (\Model{L}_{\lp_{\Model{M}}(v)} \text{ is a submodel of } \Model{L}_{\lp_{\Model{M}}(w)}) \\
	\Leftrightarrow & \Model{L}_{\lp_{\Model{M}}(w)},\lp_{\Model{M}}(w) \mmodels \Box \alpha \hspace{2ex} (=\varphi) & (\text{construction of }\Model{L}_{\lp_{\Model{M}}(w)}) ~\qedclaim
\end{mathe}

\noindent For a $\PrL_0$ instance $\langle \varphi, \Model{M}, w \rangle$ one can compute $\lp_{\Model{M}}(w)$ with the resources of \NLog (see~\cite{JT07}). 
It can be decided whether $\langle \varphi, \Model{L}_{\lp_{\Model{M}}(w)},\lp_{\Model{M}}(w) \rangle \in \fe{\PrL_0}$ with the resources of \ACi \cite{KR-presonal11}.
With Claim~\ref{claim:homomorphic_L_model} it holds that $\langle \varphi, \Model{L}_{\lp_{\Model{M}}(w)},\lp_{\Model{M}}(w) \rangle \in \fe{\PrL_0}$ if and only if $\langle \varphi, \Model{M}, w \rangle \in \fe{\PrL_0}$.
Since $\NLog \subseteq \ACi$ it holds that $\fe{\PrL_0} \in \ACi$.
\qed

%% file: 5-Conclusion.tex
%%%%%%%%%%%%%%%%%%%%%%%%%%%%%%%%%%%%%%%%%%%%%%%%%%%%%%%%%%%%%%%%%%%%%%%%%%%%%%%%%%%%%%%%%%%%%%%%%%%%%%%%%%%%
% Conclusion
%%%%%%%%%%%%%%%%%%%%%%%%%%%%%%%%%%%%%%%%%%%%%%%%%%%%%%%%%%%%%%%%%%%%%%%%%%%%%%%%%%%%%%%%%%%%%%%%%%%%%%%%%%%%

\section{Conclusion}
\label{sec:conclusion}

Now we are ready to state the \p-completeness results for the model checking problems
for intuitionistic logics and their modal companions.
Overviews are given in Figures~\ref{fig:IPC_results} and \ref{fig:ML_results}.
We start with optimal results for intuitionistic logics.

\begin{thm}\label{thm:opt-intu-results}
The model checking problem is \p-complete for 
$\FPL^{\iImpl}_1$, $\KC_2$, $\IPC_2$, and $\BPL_0$.
These results are optimal with respect to the number of variables.
\end{thm}

\proof
The upper bound from Theorem~\ref{thm:FisherLadner} carries over to all these fragments.
The \p-hardness 
for $\FPL^{\iImpl}_1$ comes from Theorem~\ref{thm:FPL1-P-hard}, 
for $\KC_2$ and $\IPC_2$ from Theorem~\ref{thm:KC2-P-hard} and 
for $\BPL_0$ from Theorem \ref{thm:BPL0-P-hard}.
The optimality for $\fe{\FPL_1^{\iImpl}}$ follows from Theorem~\ref{thm:FPL0-in-LOGCFL} where we show that $\fe{\FPL_0}$ is in \LOGCFL.
For $\fe{\IPC_2}$ and $\fe{\KC_2}$ it follows from \cite{MW-STACS11} where \ACi-completeness for $\fe{\IPC_1}$ and \NCi-completeness for $\fe{\KC_1}$ is shown.
\qed

%A completeness result for $\fe{\FPL_0}$ is still open but we give \NLog as lower bound in Lemma~\ref{lem:FPL_0-NLhard}.
For the following results the optimality is still open.

\begin{thm}\label{thm:nopt-intu-results}
The model checking problem is \p-complete for 
$\KC^{\iImpl}$, $\IPC^{\iImpl}$, and $\BPL^{\iImpl}_1$.
\end{thm}

\proof
The upper bound from Theorem~\ref{thm:FisherLadner} carries over to all these fragments.
The \p-hardness 
for $\KC^{\iImpl}$ comes from Theorem~\ref{thm:KC-P-hard}, for $\IPC^{\iImpl}$ from Corollary~\ref{cor:implIPC-BPL-p-hard}, and
for $\BPL^{\iImpl}_1$ from Corollary~\ref{cor:implBPL1-p-hard}.
\qed

It is known that the validity problem for $\IPC^{\iImpl}$ even without using $\bot$ \cite{stat79, Chagrov85, svejdar03}, for $\FPL^{\iImpl}$ and $\BPL^{\iImpl}$ \cite{Chagrov85}, and for  $\IPC_2$, $\FPL_1$, and $\BPL_0$ \cite{Rybakov06} is \PSPACE-complete.
We show for all these fragments that model checking is \p-complete. 
Even more, for the implicational fragments $\FPL^{\iImpl}_1$ and $\BPL^{\iImpl}_1$ 
with only one variable we reach \p-completeness of model checking.
Notice that no \PSPACE-hardness results for the validity problem for implicational fragments with a bounded number of variables are known.

Our \p-completeness results for $\fe{\KC^{\iImpl}}$ and $\fe{\IPC^{\iImpl}}$ hold also for the purely implicational fragments, 
i.e. $\KC^{\iImpl}$ and $\IPC^{\iImpl}$ without using $\bot$ (resp. negation).
But what happens if one bounds the number of variables in the implicational fragments?
The model checking problem for $\IPC_1^{\iImpl}$ is \NCi-complete \cite{MW-STACS11} but for $\IPC^{\iImpl}_i$ with $i>1$ 
it is open whether the complexity is below \p.
%For a given instance we compute the equivalence class of the formula and the equivalence class of the model an check whether they are compatible.
The fragments $\IPC^{\iImpl}_i$ have finitely many equivalence classes of formulas and models \cite{Urq74,RHJ10}.
This equivalence class can be obtained with the resources of \NCi,
using a straightforward extension of the Boolean formula evaluation algorithm of Buss~\cite{bus87}.
This might indicate an upper bound lower than \p for the model checking problem.
But it is not clear how hard it is to obtain the equivalence class of a given model.

Another interesting open question is the complexity of $\fe{\BPL^{\iImpl}_0}$.
We expect the \p-completeness of $\fe{\BPL^{\iImpl}_1}$ to be optimal.
But in contrast to $\IPC^{\iImpl}_i$ even $\BPL^{\iImpl}_0$ has infinitely many equivalence classes of formulas,
because $\FPL^{\iImpl}_0$ already has it~\cite{visser80}.
For $\FPL_0$, every equivalence class is represented by an implicational formula
(see proof of Theorem~\ref{thm:FPL0-in-LOGCFL}).
For $\BPL^{\iImpl}_0$, it is clear that there are more equivalence classes,
but it is open whether they can easily be represented.

%Rybakov \cite{Rybakov06} showed the \PSPACE-completeness of the validity problem for $\IPC_2$, $\FPL_1$, and $\BPL_0$.
%We show that model checking is \p-complete for $\IPC_2$, $\FPL^{\rightarrow}_1$, $\BPL^{\rightarrow}_1$, and $\BPL_0$ and that this is optimal for $\IPC_2$, $\FPL^{\rightarrow}_1$, and $\BPL_0$ with respect to the number of variables.
%
%
%We need at least one variable for the \p-completeness of $\fe{\FPL^{\rightarrow}}$ and $\fe{\BPL^{\rightarrow}}$.
%The validity problem for $\IPC^{\rightarrow}_i$ is NCi-complete because the equivalence class for a given formula can be obtained with the resources of \NCi.
%The validity problem for $\IPC_0$ is \NCi-complete and for $\IPC_1$ it is in \LOGDCFL \cite{MW-STACS11}.

%It is not hard to see that model checking for the purely implicational fragments of $\FPL_2$ and $\BPL_2$ are also \p-complete.
%We expect that this is optimal because the pureness of a fragment costs us an additional variable for simulating $\bot$.
%For \p-hardness of \fe{\BPL} we pay with variables for the restriction of connectives.
%For implicational fragments one needs at least one variable.
%As in \FPL, dropping the use of $\bot$ (resp. negation) costs a second variable.

%We show \p-completeness only for $\BPL_0$.

\begin{figure}[t]  %results intu logic

\begin{center}
\renewcommand{\arraystretch}{1.5}\vspace{1.2ex}%
\begin{tabular}{|l|cccc|}
\hline
 & \multicolumn{4}{c|}{number of variables} \\
 & unbounded & 2 & 1 & 0 \\
\hline
\BPL & \multicolumn{3}{c}{\p-complete$^{\iImpl}$} & \p-complete\\
\hline
\multirow{2}{*}{\FPL} & \multicolumn{3}{c}{\multirow{2}{*}{\p-complete$^{\iImpl}$}} & in \LOGCFL \\[-2px]
& \multicolumn{3}{c}{} & \NLog-hard$^{\iImpl}$ \\
\hline
\IPC & \p-complete$^{\iImpl}$ & \p-complete & \ACi-complete\cite{MW-STACS11} & \NCi-complete\cite{MW-STACS11} \\
\hline
\KC & \p-complete$^{\iImpl}$ & \p-complete & \multicolumn{2}{c|}{\NCi-complete\cite{MW-STACS11}} \\
\hline
\end{tabular}%
\renewcommand{\arraystretch}{0}

\caption{Complexity of the model checking problem for intuitionistic logics. \protect\\
(The $^{\iImpl}$ 
indicates that the result holds for the implicational fragment.)}
\label{fig:IPC_results}
\end{center}

\end{figure}

%%%%%%%%%%%%%%%%%%%%%%%%%%%%%%%%%%%%%%%%%%%%%%%%%%%%%%%%%%%%%%%%%%%%%%%%%%%%%%%%%%%%%%%%%%%%
% ML
%%%%%%%%%%%%%%%%%%%%%%%%%%%%%%%%%%%%%%%%%%%%%%%%%%%%%%%%%%%%%%%%%%%%%%%%%%%%%%%%%%%%%%%%%%%%

For the modal companions we conclude the following and start with the optimal results.

\begin{thm}\label{thm:opt-modal-results}
The model checking problem is \p-complete for 
$\PrL^{\rightarrow}_1$, $\Sivii_1$, $\Siv_1$, and $\Kiv_0$.
These results are optimal with respect to the number of variables.
\end{thm}

\proof
For all these fragments the upper bound comes from Theorem~\ref{thm:FisherLadner}.
The \p-hardness 
for $\PrL^{\rightarrow}_1$, and $\Kiv_0$ comes from Theorem~\ref{thm:lowerbounds_modalL},
for $\Sivii_1$ from Theorem \ref{thm:S4.2_1-P-hard}, and
for $\Siv_1$ from Corollary~\ref{cor:S4_1-P-hard}.
The optimality for $\fe{\PrL_1^{\rightarrow}}$ follows from Theorem~\ref{thm:PrL0-in-AC1} where we show $\fe{\PrL_0} \in \ACi$.
For $\fe{\Sivii_1}$ and $\fe{\Siv_1}$ it follows from Lemma~\ref{lem:S40_in_NC1} where \NCi-completeness for $\fe{\Sivii_0}$ and $\fe{\Siv_0}$ is shown.
\qed

%As lower bound for \fe{\PrL_0} we give \NLog (Lemma~\ref{lem:PrL_0-NLhard}).
%A completeness result is still open.
Notice that $\fe{\IPC_1}$ and $\fe{\KC_1}$ are the only cases where model checking
for intuitionistic logics is easier than for its modal companions $\fe{\Siv_1}$ and $\fe{\Sivii_1}$.

For the following results the optimality is still open.

\begin{thm}\label{thm:nopt-modal-results}
The model checking problem is \p-complete for 
$\Sivii^{\rightarrow}$, $\Siv^{\rightarrow}$, and $\Kiv^{\rightarrow}_1$.
\end{thm}

\proof
For all these fragments the upper bound comes from Theorem~\ref{thm:FisherLadner} and the \p-hardness comes from Theorem~\ref{thm:lowerbounds_modalL}.
\qed

Completeness results for $\fe{\Sivii^{\rightarrow}}$ and $\fe{\Siv^{\rightarrow}}$ with a bounded number of variables 
and for $\Kiv^{\rightarrow}_0$ are still open.

\begin{figure}[t]  %results modal logic

\begin{center}
\renewcommand{\arraystretch}{1.5}\vspace{1.2ex}%
\begin{tabular}{|l|ccc|}
\hline
 & \multicolumn{3}{c|}{number of variables} \\
 & unbounded &  1 & 0 \\
\hline
\KK & \multicolumn{3}{c|}{\p-complete$^{\rightarrow}$} \\
\hline
\Kiv & \multicolumn{2}{c}{\p-complete$^{\rightarrow}$} & \p-complete\\
\hline
\multirow{2}{*}{\PrL} & \multicolumn{2}{c}{\multirow{2}{*}{\p-complete$^{\rightarrow}$}} & in \ACi \\[-2px]
& & & \NLog-hard$^{\rightarrow}$ \\
\hline
\Siv & \p-complete$^{\rightarrow}$ & \p-complete & \NCi-complete \\
\hline
\Sivii & \p-complete$^{\rightarrow}$ & \p-complete & \NCi-complete \\
\hline
\end{tabular}%
\renewcommand{\arraystretch}{0}

\caption{Complexity of the model checking problem for the modal companions. \protect\\
(The $^{\rightarrow}$ indicates that the result holds for the strictly implicational fragment.)}
\label{fig:ML_results}
\end{center}
\end{figure}

Another semantics for intuitionistic logics is the class of finite trees that are reflexive and transitive.
This is a subclass of the intuitionistic Kripke models we used and also sound and complete for \IPC.
It is open whether the model checking problem for $\IPC$ over this tree-semantics is \p-hard or below \p, 
and it also remains open for the other \p-complete model checking problems of this work.

%% file: KC2_P-haerte.tex
%\begin{thm}\label{thm:KC2-P-hard}
%The model checking problem for $\KC_2$ and for $\IPC_2$ is \p-hard.
%\end{thm}

\proof
We show $\fe{\IPC^{\iImpl}}\mlogred\fe{\KC_2}$.
Then \p-hardness for $\fe{\KC_2}$ and $\fe{\IPC_2}$ follows from Corollary~\ref{cor:implIPC-BPL-p-hard}.
The construction is similar to the one given by Rybakov \cite[Theorem 4]{Rybakov06} for the \PSPACE-completeness of the validity problem for $\IPC_2$.
First of all we construct formulas with two variables which can be used for replacing the variables in arbitrary \Klasse{IL} formulas.
We call them \textit{replacement formulas}.
Then we give \textit{generic models} which have for every replacement formula a unique maximal refuting state\footnote{
	In $\Model{M} = (W,R,\xi)$ the state $w$ \textit{refutes} $\varphi$ if $\Model{M},w \not \imodels \varphi$.
	A state $w \in W$ is a \textit{maximal refuting state} of $\varphi$ if for all $v \in W \setminus \{w\}$ with $(w,v) \in R$ it holds that $\Model{M},v \imodels \varphi$.}.
For a given instance of $\fe{\IPC}^{\iImpl}$ we transform the formula by replacing the variables with the replacement formulas and as model we take the union of the given model and a suitable generic model.
%With Claim \ref{claim:IPC2redu} we show that this is a correct reduction. 
This union eventually is a $\KC_2$ model.

The construction---especially the base of the inductive definition of the replacement formulas---is very technical.
Let $p$ and $q$ be the variables used in $\KC_2$. 
Figure~\ref{fig:IPC2_P-haerte_Modell_1} shows the top of the generic model.
There, one can see in which states the variables $p$ and $q$ are satisfied
(due to the valuation function of the model),
and which are the maximal refuting states of the formulas to be defined in the sequel.
The essential idea is that every replacement formula has exactly one state that is its maximal refuting state.
%The essential idea is that every state (except the top most state $c$)
%is the maximal refuting state of exactly one of the replacement formulas.
%
\input{Bilder-IPC2_P-haerte_Modell_1}

\textit{Construction of the replacement formulas.}
The following formulas are the base for the inductive definition of the replacement formulas. 

\vspace{1.7ex}
\begin{tabular}{lcl@{\hspace{10mm}}lcl}
$\delta_1$ & $:=$ & $p \iImpl q$ & $\delta_2$ & $:=$ & $q \iImpl p$ \hspace{12ex} $\delta_3 \hspace{1.7ex}:=\hspace{1.7ex} p \vee q$ \\[10px]
%$\delta_3$ & $:=$ & $(\delta_1 \wedge \delta_2) \rightarrow (p \wedge q)$ \\[10px]
% $\delta_2$ & $:=$ & $q \rightarrow p$ & & & \\[10px]
$\eps_1$ & $:=$ & $\delta_2 \iImpl (\delta_1 \vee \delta_3)$ & 
		$\eps_3$ & $:=$ & $\delta_1 \iImpl (\delta_2 \vee \delta_3)$ \\[3px]
$\eps_2$ & $:=$ & $\delta_3 \iImpl (\delta_1 \vee \delta_2)$ & 
		$\eps_4$ & $:=$ & $(\eps_1 \wedge \eps_2 \wedge \eps_3) \iImpl (\delta_1 \vee \delta_2 \vee \delta_3)$
\end{tabular} 
\vspace{1.7ex}\newline
Using these formulas, the first replacement formulas can be defined as follows.

\vspace{1.7ex}
\begin{tabular}{lcl@{\hspace{15mm}}lcl}
$\alpha_1^1$ & $:=$ & $(\eps_1 \wedge \eps_2) \iImpl (\eps_3 \vee \eps_4)$ & 
		$\beta_1^1$ & $:=$ & $(\eps_2 \wedge \eps_3) \iImpl (\eps_1 \vee \eps_4)$ \\[3px]
$\alpha_2^1$ & $:=$ & $(\eps_1 \wedge \eps_3) \iImpl (\eps_2 \vee \eps_4)$ & 
		$\beta_2^1$ & $:=$ & $(\eps_2 \wedge \eps_4) \iImpl (\eps_1 \vee \eps_3)$ \\[3px]
$\alpha_3^1$ & $:=$ & $(\eps_1 \wedge \eps_4) \iImpl (\eps_2 \vee \eps_3)$ & 
		$\beta_3^1$ & $:=$ & $(\eps_3 \wedge \eps_4) \iImpl (\eps_1 \vee \eps_2)$ \\
\end{tabular} 
\vspace{1.7ex}\newline
We call the upper index the \textit{level}.
The formulas on the next levels will be defined inductively.
First we define $n_1:=3$ and $n_{k+1}:=|P_k|$ where $P_k:=\{(x,y)\mid 2 \leq x,y \leq n_k\}$. 
With induction on $k$ one can show that  $|P_k|=(n_k-1)^2$.
On level $k$ we define $\alpha_i^k$ and $\beta_i^k$ for $i=1,2,\dots,n_k$.
%For a fixed $k \geq 1$ suppose that the formulas $\alpha_i^k$ and $\beta_i^k$ with $1\leq i\leq n_k$ of level $k$ are already defined. 
For the step from level $k$ to level $k+1$ we need an encoding $\langle\cdotp,\cdot\rangle_k$ from $P_k$ to $\{1,2,\ldots,(n_k-1)^2\}$ that is easy to compute and easy to decode.
%Note that $|P_k|=(n_k-1)^2$ hence it is easy to give such a bijective mapping $\langle\cdotp,\cdot\rangle_k$.
For example one can use the following:
$\langle\cdotp,\cdot\rangle_k$ maps $(i,j)$ to $(j-1)+(n_k-1)\cdot(i-2)$ for $2 \leq i,j \leq n_k$.
%\footnote{
%For example one can define a strict linear order $\lhd$ on pairs such that $(i,j) \lhd (l,m)$ iff $i+j < l+m$ or $i+j = l+m$ and $i<l$.
%This order implies an enumeration of those pairs that can be used as $g_k$, i.e. $g_2(1) = (2,2)$, $g_2(2)=(2,3)$, $g_2(3)=(3,2)$ and so on.
%}
For $k\geq 1$ the inductive definition is as follows.
Let $i,j \in \{2,3,\dots,n_k\}$. 

\begin{mathe}
	\alpha_{\langle i,j\rangle_k}^{k+1} & := & \alpha_1^k \iImpl \klauf \beta_1^k \vee \alpha_i^k \vee \beta_j^k \klzu \\
	\beta_{\langle i,j\rangle_k}^{k+1}  & := & \beta_1^k \iImpl  \klauf \alpha_1^k \vee \alpha_i^k \vee \beta_j^k \klzu 	
\end{mathe}

\textit{Construction of the generic models.}
For $t\geq1$ we define the generic models $\Model{M}^S_t=(W^S_t,R^S_t,\xi^S)$.
% that has for every formula $\alpha_i^k$ and $\beta_i^k$ a unique maximal state that refutes the formula. 

\begin{mathe}
	W_0 & := & \gklauf c,d_1,d_2,d_3,e_1,e_2,e_3,e_4\gklzu \\
	W_k & := & \gklauf a_i^k,b_i^k \mid 1 \leq i \leq n_k \gklzu ~~\text{ for } 1 \leq k \leq t \\
	W^S_t & := & {\displaystyle \bigcup_{l=0}^t} \hspace{3px} W_l
\end{mathe}
In the following we give $R^S_t$.
%Let $a_i^k$ and $b_i^k$ be states in level $k$.
The accessibility relation $R_{\Top}$ of the first layers is shown in Figure \ref{fig:IPC2_P-haerte_Modell_1}.
(Certainly we use the transitive and reflexive closure of the depicted edges.)
For states from level $\geq2$ the accessibility relation will be defined as follows.
Let $1 \leq k \leq t-1$. 

%
%\begin{mathe}
%	R_0 & := & \{(d_1,c),(d_2,c),(d_3,c),(a_1^0,d_1),(a_2^0,d_1),(b_2^0,d_1), \\[-5px]
%				&			 & \hspace{1.2ex} (a_2^0,d_2),(b_1^0,d_2),(b_2^0,d_2),(a_1^0,d_3),(b_1^0,d_3),(b_2^0,d_3)\} \\
%	R_1^a & := & \{(a_1^1,b_1^0),(a_1^1,b_2^0),(a_2^1,a_2^0),(a_2^1,b_2^0),(a_3^1,a_2^0),(a_3^1,b_1^0)\} \\
%	R_1^a & := & \{(b_1^1,a_1^0),(b_1^1,b_2^0),(b_2^1,a_1^0),(b_2^1,b_1^0),(b_3^1,a_1^0),(b_3^1,a_2^0)\} 
%\end{mathe}

%The structure constructed up to now is shown in Figure \ref{fig:IPC2_P-haerte_Modell_1}.

%\input{Bilder-IPC2_P-haerte_Modell_1}
%
\begin{mathe}
	R_{k+1}^a & := & \gklauf(a_{\langle i,j \rangle_k}^{k+1},b_1^k),(a_{\langle i,j \rangle_k}^{k+1},a_i^k),(a_{\langle i,j \rangle_k}^{k+1},b_j^k) \mid 2 \leq i,j \leq n_k \gklzu \\
	R_{k+1}^b & := & \gklauf(b_{\langle i,j \rangle_k}^{k+1},a_1^k),(b_{\langle i,j \rangle_k}^{k+1},a_i^k),(b_{\langle i,j \rangle_k}^{k+1},b_j^k)  \mid 2 \leq i,j \leq n_k \gklzu \\
	R' 				& := & R_{\Top} \hspace{1px} \cup \hspace{3px} {\displaystyle \bigcup_{l=2}^{t}} \hspace{5px} (R_l^a \cup R_l^b).
\end{mathe}
In order to make the accessibility relation transitive, we add pseudo-transitive edges.
Every state in a level is connected to every state at least two levels below. 

\begin{mathe}
	T_k	& := & W_k \times \Klauf {\displaystyle \bigcup_{l=0}^{k-2}} \hspace{5px} W_l \Klzu ~~\text{ for } k \geq 2
\end{mathe}
$T$ is the union of all pseudo-transitive edges.

\begin{mathe}
	T & := & {\displaystyle \bigcup_{l=2}^{t}} \hspace{5px} T_l
\end{mathe}
We define the accessibility relation $R^S_t$ as follows.

\begin{mathe}
	R^S_t \text{ is the reflexive closure of } T \cup R'. & &
\end{mathe}
Figure \ref{fig:IPC2_P-haerte_Modell_2} shows a cutout of $\Model{M}^S_t$.
\input{Bilder-IPC2_P-haerte_Modell_2}
The valuation function $\xi^S$ is defined as follows (see Figure~\ref{fig:IPC2_P-haerte_Modell_1}).

\begin{mathe}
	\xi^S(p) & := & \{c,d_1\} \\
	\xi^S(q) & := & \{c,d_2\}
\end{mathe}										
%											
%Now we can speak about the connection between the replacement formulas and the generic model an claim the following (see~\cite[Lemma 5]{Rybakov06}).
The goal of the construction is that $\alpha_i^k$ (resp. $\beta_i^k$) is not satisfied exactly in the states that see $a_i^k$ (resp. $b_i^k$).
\begin{claim} \label{claim:KC2standardModel_prop}
Let $w$ be a state of $\Model{M}^S_t$. 
Then for all $1 \leq k \leq t$ and $i\leq n_k$ it holds that 	$\Model{M}^S_t,w \not\imodels \alpha_i^k  \Leftrightarrow  (w,a_i^k) \in R^S_t$ and $\Model{M}^S_t,w \not\imodels \beta_i^k \Leftrightarrow (w,b_i^k) \in R^S_t$. %~\qedclaim
\end{claim}

The proof can be proceeded by an induction on $k$ (similar as~\cite[Lemma 5]{Rybakov06}).
Hence for every formula $\alpha_i^k$ and $\beta_i^k$ exists a unique maximal state in $\Model{M}^S_t$ that refutes this formula.

\textit{Reduction from $\IPC^{\iImpl}$ model checking problem.}
For a given instance $\langle \varphi,\Model{M},w \rangle$ of $\fe{\IPC^{\iImpl}}$ we show how to translate $\Model{M}$ and $\varphi$ into $\Model{M}^2$---a model over two variables---and $\varphi^2$---a formula with two variables. 
Let $\varphi$ be a formula with variables $v_1,v_2,\dots,v_m$ and $\Model{M}=(W,R,\xi)$ a model.
%We assume that $\xi(x)=\emptyset$ if $x \notin \{v_1,v_2,\dots,v_m\}$ and $W^S \cap W = \emptyset$.
We choose the smallest $k>1$ such that $n_k > m$.
%For the translation $\varphi^2$ of $\varphi$ we choose the smallest $k$ with $n_k > m$.
%(For technical reasons $k$ has to be greater than 5.)
To define $\varphi^2$ we replace every occurrence of $v_i$ in $\varphi$ by $\alpha_i^k \vee \beta_i^k$.

\begin{mathe}
	\varphi^2 & := & \varphi[v_1/\alpha_1^k \vee \beta_1^k][v_2/\alpha_2^k \vee \beta_2^k]\ldots[v_m/\alpha_m^k \vee \beta_m^k]
\end{mathe}
%
%With \cite[Lemma 6]{Rybakov06} it follows that $|\varphi^2| \geq c \cdot |\varphi|^2$ (for a constant $c$ independent of $\varphi$) and $\varphi^2$ can be constructed with a logspace algorithm.
Since $k \leq 1+\log(m)$ one can construct $\varphi^2$ in logspace.
%We define now $\Model{M}^S_k = (W^S_k,R^S_k,\xi^S)$.
%This is the suitable part of the generic model $\Model{M}^S$ that consists of the first $k+1$ levels of $\Model{M}^S$.
%
%\begin{mathe}
%	W^S_k & := & {\displaystyle \bigcup_{l=0}^{k}} \hspace{5px} W_l \\
%	R^S_k & := & R^S \text{ restricted to } W^S_k \\
%\end{mathe}
%
We build the translation $\Model{M}^2=(W^2,R^2,\xi^2)$ as a union of $\Model{M}$ and $\Model{M}^S_k$. 

\begin{mathe}
	W^2 & := & W \cup W^S_k
\end{mathe}
The accessibility relation $R^2$ is constructed such that if $w \notin \xi(v_i)$, then $(w,a_i^k) \in R^2$ and $(w,b_i^k) \in R^2$.
Hence $w$ refutes $\alpha_i^k \vee \beta_i^k$---the translation of $v_i$.

\begin{mathe}
	R_{\xi} & := & \gklauf(w,a_i^k),(w,b_i^k) \mid w \in W \setminus \xi(v_i)\gklzu \cup \gklauf(w,a_{m+1}^k),(w,b_{m+1}^k) \mid w \in W\gklzu
\end{mathe}
In order to make $R^2$ transitive and give a logspace computable construction we connect every state of $\Model{M}$ with every state in $\Model{M}^S$ on level $k-1$ and below. 

\begin{mathe}
	R_{\emph{trans}} & := & W \times {\displaystyle \bigcup_{l=0}^{k-1}} \hspace{5px} W_l
\end{mathe}	
We define the accessibility relation $R^2$ as follows.

\begin{mathe}
	R^2 \text{ is the reflexive closure of } R^S_k \cup R \cup R_{\xi} \cup R_{\emph{trans}}. & &
\end{mathe}
%The accessibility relation $R^2$ is the reflexive closure of $R^S \cup R \cup R_{\xi} \cup R_{\emph{trans}}$.
As valuation function we use 

\begin{mathe}
	\xi^2 & := & \xi^S.
\end{mathe}
The valuation function $\xi$ of $\Model{M}$ is simulated by the edges between $\Model{M}$ and $\Model{M}^S$ from $R_{\xi}$.
The model $\Model{M}^2$ is a $\KC_2$ model because for every state $u\in W^2$ it holds that $(u,c) \in R^2$. 
%Now we can show that this is a correct reduction.
(The proof of the following claim bases on the proof of \cite[Lemma 7]{Rybakov06}.)
\begin{claim}\label{claim:IPC2redu}
For all $w \in W$ it holds that $\Model{M},w \imodels \varphi$ if and only if $\Model{M}^2,w \imodels \varphi^2$.
\end{claim}

\noindent
\emph{Proof of Claim.}
We prove this by induction on the construction of $\varphi$.
For the initial step let $\varphi = v_l$ be a variable with $1\leq l\leq m$, hence $\varphi^2 = \alpha_l^k \vee \beta_l^k$.
If $\Model{M},w \not\imodels v_l$, then $w$ is via $R^2$ (resp. $R_{\xi}$) connected to $a_l^k$ and $b_l^k$.
With Claim~\ref{claim:KC2standardModel_prop} it follows that $\Model{M}^2,w \not\imodels \alpha_l^k \vee \beta_l^k$. 
Now assume that $\Model{M}^2,w \not\imodels \alpha_l^k \vee \beta_l^k$ with $w \in W$.
Since 

\begin{mathe}
	\alpha_l^k & = & \alpha_1^{k-1} \iImpl \klauf \beta_1^{k-1} \vee \alpha_i^{k-1} \vee \beta_j^{k-1} \klzu \text{\hspace{1.8ex} and} \\
	\beta_l^k & = & \beta_1^{k-1} \iImpl \klauf \alpha_1^{k-1} \vee \alpha_i^{k-1} \vee \beta_j^{k-1} \klzu
\end{mathe}
for $\langle i,j \rangle_{k-1}=l$ it holds that there are some states $w',w'' \in W^2$ with $(w,w') \in R^2$ and $(w,w'') \in R^2$ and 

\begin{mathe}
	\Model{M}^2,w' \hspace{0.47ex} \imodels \alpha_1^{k-1} & \text{and} & 
					\Model{M}^2,w' \hspace{0.47ex} \not\imodels \beta_1^{k-1} \vee \alpha_i^{k-1} \vee \beta_j^{k-1} \text{\hspace{1.8ex} and} \\
	\Model{M}^2,w'' \imodels \beta_1^{k-1} & \text{and} &
					\Model{M}^2,w'' \not\imodels \alpha_1^{k-1} \vee \alpha_i^{k-1} \vee \beta_j^{k-1}.
\end{mathe}
From Claim~\ref{claim:KC2standardModel_prop} follows that $\Model{M}^2,a_{m+1}^k \not\imodels \beta_1^{k-1}$ and $\Model{M}^2,b_{m+1}^k \not\imodels \alpha_1^{k-1}$.
Hence it follows for every $u \in W$ that $\Model{M}^2,u \not\imodels \alpha_1^{k-1}$ and $\Model{M}^2,u \not\imodels \beta_1^{k-1}$ because $(u,a_{m+1}^k) \in R^2$ and $(u,b_{m+1}^k) \in R^2$.
Therefore $w',w'' \in W^S_k$.
Furthermore note that $w'$ and $w''$ are in level $k$ of $\Model{M}^S_k$ because $w'$ refutes $\alpha_l^k$ and $w''$ refutes $\beta_l^k$
and with Claim~\ref{claim:KC2standardModel_prop} it follows that $w'=a_l^k$ and $w''=b_l^k$.
%Exactly one state in level $k$ refutes $\alpha_l^k$ and one $\beta_l^k$, hence $w'=a_l^k$ and $w''=b_l^k$.
From $(w,a_l^k) \in R^2$, $(w,b_l^k) \in R^2$, and the construction of $R^2$ it follows that $w \notin \xi(v_l)$.
%From $wR^2a_l^k$ and the construction of $\Model{M}^2$ it follows that there is a state $u \in W$ with $wRu$, $uR^*a_l^k$ and $u \notin \xi(v_l)$.
%Because of the monotonicity of intuitionistic models it holds that $w \notin \xi(v_l)$.
Hence $\Model{M},w \not\imodels v_l$.

For the induction step let $\varphi = \gamma \star \delta$ with $\star \in \{\wedge, \vee, \iImpl\}$.
We show that $\Model{M}^2,w \imodels (\gamma \star \delta)^2$ if and only if $\Model{M},w \imodels \gamma \star \delta$.
(Note that $(\gamma \star \delta)^2 = \gamma^2 \star \delta^2$.)
For the cases that $\star = \wedge$ and $\star = \vee$ this follows directly from the definition of the satisfaction relation $\imodels$.
Now consider $\varphi = \gamma \iImpl \delta$ and $\Model{M},w \not\imodels \varphi$.
Then there is some state $w'\in W$ with $\Model{M},w' \imodels \gamma$ and $\Model{M},w' \not\imodels \delta$. 
By induction hypothesis it follows that $\Model{M}^2,w' \imodels \gamma^2$ and $\Model{M}^2,w' \not\imodels \delta^2$.
Hence $\Model{M}^2,w \not\imodels \varphi^2$. 
For the other proof direction let $w \in W$ with $\Model{M}^2,w \not\imodels \varphi^2$. 
Then there is a $w'\in W^2$ with $(w,w')\in R^2$ and $\Model{M}^2,w' \imodels \gamma^2$ and $\Model{M}^2,w' \not\imodels \delta^2$.
The formulas $\alpha_1^k \vee \beta_1^k, \alpha_2^k \vee \beta_2^k, \dots, \alpha_m^k \vee \beta_m^k$ are satisfied in every state of level $k$ and below, because every state in level $k$ refutes exactly one $\alpha_i^k$ respectively $\beta_i^k$ formula. 
(The states below level $k$ satisfy all formulas on level $k$.)
In $\delta^2$ every variable $v_i$ from $\delta$ is replaced by the disjunction of an $\alpha_i^k \vee \beta_i^k$.
Hence $\delta^2$ is satisfied in every state in $W^S_k$.
In order that $w'$ refutes $\delta^2$, it holds that $w'\in W$.
By induction hypothesis we obtain that $\Model{M},w' \imodels \gamma$ and $\Model{M},w' \not\imodels \delta$.
From $w,w' \in W$ and $(w,w') \in R^2$ it follows that $(w,w') \in R$.
Hence $\Model{M},w \not\imodels \varphi$.  ~\qedclaim

The reduction function is the mapping

\begin{mathe}
	\langle \varphi,\Model{M},w \rangle & \longmapsto & \langle \varphi^2,\Model{M}^2,w \rangle
\end{mathe}
where $\langle \varphi,\Model{M},w \rangle$ is an instance of $\fe{\IPC^{\iImpl}}$.
Claim \ref{claim:IPC2redu} shows that $\Model{M},w \imodels \varphi$ if and only if $\Model{M}^2,w \imodels \varphi^2$.
It follows directly from the construction that this is a logspace reduction.

\qed

%% file: Bilder-IPC2_P-haerte_Modell_1.tex
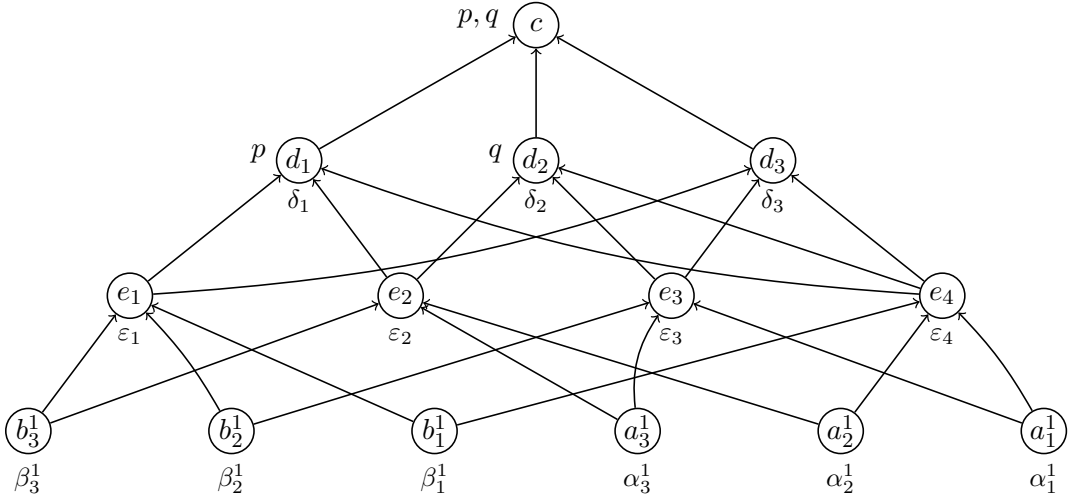
\begin{figure}[ht]
\hrulefill
\vspace{1ex}

\begin{tikzpicture}[
 scale=0.9,
 -,
 auto,
 node distance=1.2cm,
 semithick,
 state/.style={style=circle, draw=black, minimum size=6mm, inner sep=0mm},
 empty/.style={style=circle, draw=white, minimum size=0mm, inner sep=0mm},
 txt/.style={style=rectangle},]

 \node[state] (c) at (5.5,6) {$c$};
 \node[txt] at (4.65,6.1) {$p,q$};
 
 \node[state] (d1) at (2,4) {$d_1$};
 \node[txt] at (1.4,4.1) {$p$};
 \node[state] (d2) at (5.5,4) {$d_2$};
 \node[txt] at (4.9,4.1) {$q$};
 \node[state] (d3) at (9,4) {$d_3$};
 
 \node[state] (a10) at (-0.5,2) {$e_1$};
 \node[state] (a20) at (3.5,2) {$e_2$};
 \node[state] (b10) at (7.5,2) {$e_3$};
 \node[state] (b20) at (11.5,2) {$e_4$};

 \node[state] (b31) at (-2,0) {$b_3^1$};
 \node[state] (b21) at (1,0) {$b_2^1$};
 \node[state] (b11) at (4,0) {$b_1^1$};
 \node[state] (a31) at (7,0) {$a_3^1$};
 \node[state] (a21) at (10,0) {$a_2^1$};
 \node[state] (a11) at (13,0) {$a_1^1$};

\begin{small}
 \node[txt] at (2,3.4) {$\delta_1$};
 \node[txt] at (5.5,3.4) {$\delta_2$}; 
 \node[txt] at (9,3.4) {$\delta_3$};

 \node[txt] at (-0.5,1.4) {$\eps_1$};
 \node[txt] at (3.5,1.4) {$\eps_2$};
 \node[txt] at (7.5,1.4) {$\eps_3$};
 \node[txt] at (11.5,1.4) {$\eps_4$};

 \node[txt] at (-2,-0.7) {$\beta_3^1$};
 \node[txt] at (1,-0.7) {$\beta_2^1$};
 \node[txt] at (4,-0.7) {$\beta_1^1$};
 \node[txt] at (7,-0.7) {$\alpha_3^1$};
 \node[txt] at (10,-0.7) {$\alpha_2^1$};
 \node[txt] at (13,-0.7) {$\alpha_1^1$};

\end{small}

 \path
 (d1) edge[->] (c)
 (d2) edge[->] (c)
 (d3) edge[->] (c)
 
 (a10) edge[->] (d1)
 (a20) edge[->] (d1)
 (a20) edge[->] (d2)
 (b10) edge[->] (d2)
 (a10) edge[->,bend right=8] (d3)
 (b10) edge[->] (d3)
 (b20) edge[->,bend left=8] (d1)
 (b20) edge[->] (d2)
 (b20) edge[->] (d3)
 
 (b31) edge[->] (a10)
 (b31) edge[->] (a20)
 (b21) edge[->,bend right=7] (a10)
 (b21) edge[->] (b10)
 (b11) edge[->] (a10)
 (b11) edge[->] (b20)
 (a31) edge[->] (a20)
 (a31) edge[->,bend left=20] (b10)
 (a21) edge[->] (a20)
 (a21) edge[->] (b20)
 (a11) edge[->] (b10)
 (a11) edge[->,bend right=7] (b20)

 ;

\end{tikzpicture}
\caption{This is the top of the generic model with states from $W_0 \cup W_1$.
Every state is labelled below with the formula that it maximally refutes. 
(Transitive and reflexive edges are not depicted.)
}
\label{fig:IPC2_P-haerte_Modell_1}
\hrulefill
\end{figure}

%% file: Bilder-IPC2_P-haerte_Modell_2.tex
\begin{figure}[ht]
\hrulefill
\vspace{1ex}

\begin{tikzpicture}[
 scale=0.9,
 -,
 auto,
 node distance=1.2cm,
 semithick,
 state/.style={style=circle, draw=black, minimum size=8.5mm, inner sep=0mm},
 empty/.style={style=circle, draw=white, minimum size=0mm, inner sep=0mm},
 txt/.style={style=rectangle},]

\begin{small}
 \node[state] (a-k+1-s) at (3,0) {$a_s^{k}$};
 \node[state] (b-k+1-s) at (9,0) {$b_s^{k}$};
 
 \node[state] (b-k-1) at (2,2) {$b_1^{k-1}$};
 \node[state] (b-k-j) at (4,2) {$b_j^{k-1}$};
 \node[state] (a-k-i) at (8,2) {$a_i^{k-1}$};
 \node[state] (a-k-1) at (10,2) {$a_1^{k-1}$};
  
 \node[state] (a-k-1-l) at (5,5) {$a_{\ell}^{k-2}$};
 \node[state] (b-k-1-m) at (8,5) {$b_h^{k-2}$};
 \node[state] (b-k-1-1) at (11,5) {$b_1^{k-2}$};
\end{small} 
 
 \node[empty] (1-1) at (1.6,3.2) {};
 \node[empty] (1-2) at (2,3.2) {};
 \node[empty] (1-3) at (2.4,3.2) {};
 \node[empty] (1-4) at (3.6,3.2) {};
 \node[empty] (1-5) at (4,3.2) {};
 \node[empty] (1-6) at (4.4,3.2) {};
 \node[empty] (1-7) at (9.6,3.2) {};
 \node[empty] (1-8) at (10,3.2) {};
 \node[empty] (1-9) at (10.4,3.2) {};

 \node[txt] at (2,3.3) {$\dots$};
 \node[txt] at (4,3.3) {$\dots$};
 \node[txt] at (10,3.3) {$\dots$};
 
 \path
 (a-k+1-s) edge[->] (b-k-1)
 (a-k+1-s) edge[->] (a-k-i)
 (a-k+1-s) edge[->] (b-k-j)
 (b-k+1-s) edge[->] (a-k-1)
 (b-k+1-s) edge[->] (b-k-j)
 (b-k+1-s) edge[->] (a-k-i)

 (a-k-i) edge[->] (b-k-1-1)
 (a-k-i) edge[->] (a-k-1-l)
 (a-k-i) edge[->] (b-k-1-m)

 (a-k+1-s) edge[->, dashed, draw = black!40] (b-k-1-1)
 (a-k+1-s) edge[->, dashed, draw = black!40, bend left=22] (a-k-1-l)
 (a-k+1-s) edge[->, dashed, draw = black!40] (b-k-1-m)
 (b-k+1-s) edge[->, dashed, draw = black!40, bend right=14] (b-k-1-m)
 (b-k+1-s) edge[->, dashed, draw = black!40, bend left=10] (a-k-1-l)
 (b-k+1-s) edge[->, dashed, draw = black!40, bend right=35] (b-k-1-1)

 (b-k-1) edge[-] (1-1)
 (b-k-1) edge[-] (1-2)
 (b-k-1) edge[-] (1-3)
 (b-k-j) edge[-] (1-4)
 (b-k-j) edge[-] (1-5)
 (b-k-j) edge[-] (1-6)
 (a-k-1) edge[-] (1-7)
 (a-k-1) edge[-] (1-8)
 (a-k-1) edge[-] (1-9)
;
 
\end{tikzpicture}
\caption{This is a cutout of the levels $k-2$, $k-1$ and $k$ of a generic model $\Model{M}^S_t$ where $s=\langle i,j \rangle_{k-1}$, $i=\langle \ell,h \rangle_{k-2}$, and $k \leq t$.
The dashed grey edges are the pseudo-transitive edges. 
As we show in Claim \ref{claim:KC2standardModel_prop} for example $a_i^{k-1}$ is the maximal refuting state for $\alpha_i^{k-1}$.
(Reflexive edges are not depicted.)}
\label{fig:IPC2_P-haerte_Modell_2}
\hrulefill

\end{figure}
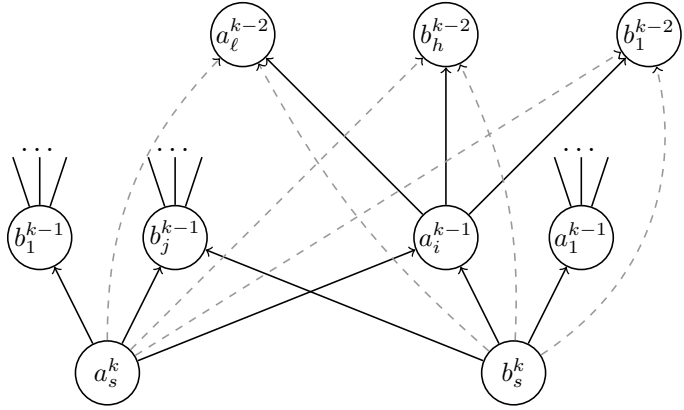 